\documentclass[11pt]{article}
\usepackage{tikz}
\usepackage{multicol}
\usepackage{booktabs}
\usepackage{bm}
\usepackage{mdframed}
\usepackage{color}
\usepackage{authblk}

\usepackage{epsfig,subfigure}
\usepackage{amssymb, amsmath}
\usepackage{color}
\usepackage{graphicx}
\usepackage{epstopdf}
\usepackage{palatino}

\usepackage{graphicx}
\usepackage{color}

\usepackage{arydshln}

\usepackage{amssymb}
\usepackage{amsmath,amsfonts,amssymb}
\usepackage{verbatim}
\usepackage{stfloats}
\usepackage[bookmarks=false]{}
\usepackage{cite}

\newtheorem{theorem}{Theorem}
\newtheorem{corollary}{Corollary}

\newtheorem{lemma}{Lemma}

\newcommand{\vect}[1]{\overrightarrow{#1}}

\begin{document}
\title{Cross Subspace Alignment and the Asymptotic  Capacity of $X$--Secure $T$--Private Information Retrieval}
\author{Zhuqing Jia$^\dagger$, Hua Sun$^*$ and Syed A. Jafar$^\dagger$}
\affil{$^\dagger$Center for Pervasive Communications and Computing (CPCC), UC Irvine\\
$^*$Department of Electrical Engineering, University of North Texas\\
Email: \{zhuqingj\}@uci.edu, \{hua.sun\}@unt.edu,  \{syed\}@uci.edu}
\date{}

\maketitle

\begin{abstract} 
$X$--secure and $T$--private information retrieval (XSTPIR) is a form of  private information retrieval  where data security is guaranteed against collusion among up to $X$ servers and the user's privacy is guaranteed against collusion among up to $T$ servers. The capacity of XSTPIR is characterized for arbitrary number of servers $N$, and arbitrary security and privacy thresholds $X$ and $T$,  in the limit as the number of messages $K\rightarrow\infty$. Capacity is also characterized for any number of messages if either $N=3, X=T=1$ or if $N\leq X+T$.  Insights are drawn from these results, about aligning versus decoding noise, dependence of PIR rate on field size, and robustness to  symmetric security constraints. In particular, the idea of cross subspace alignment, i.e., introducing a subspace dependence between Reed-Solomon code parameters, emerges as the  optimal way to align undesired terms while keeping desired terms resolvable. 
\end{abstract}
\section{Introduction}

Motivated by the importance of security and privacy  in the era of big data and distributed storage, in this work we explore the information theoretic capacity of  private information retrieval (PIR) in a secure distributed storage system. Specifically, our focus is on the $X$-secure and $T$-private information retrieval problem (XSTPIR).  A PIR scheme is said to be $T$-private if it allows a user to retrieve a desired message from a database of $K$ messages stored at $N$ distributed servers, without revealing any information about the identity of the desired message to any group of up to $T$ colluding servers. Similarly, a distributed storage scheme is said to be $X$-secure\footnote{In other words, everything that is stored at any $X$ servers must be independent of the $K$ messages. Besides $X$-security, no other constraints are imposed on the storage. The storage and the PIR scheme are jointly optimized to maximize the capacity of XSTPIR.} if it guarantees that any group of up to $X$ colluding servers learn nothing about the stored data. The $T$ and $X$ parameters may be chosen arbitrarily depending on the relative importance of security and privacy for any given application. 

The rate of a PIR scheme is the ratio of the number of bits retrieved by the user to the  total number of bits downloaded from all servers.  The supremum of achievable rates  is called the capacity of PIR. The capacity of the basic PIR setting was found in \cite{Sun_Jafar_PIR} to be  
\begin{align}
C_{\text{PIR}}(N,K)&=(1+1/N+1/N^2+\dots+1/N^{K-1})^{-1}.
\end{align}
The result was generalized subsequently in \cite{Sun_Jafar_TPIR}  to the $T$-PIR setting,  as 
\begin{equation}
C_{\mbox{\tiny TPIR}}(N, K, T)=
\left\{
\begin{array}{ll}
\left(1+T/N+T^2/N^2+\dots+T^{K-1}/N^{K-1}\right)^{-1},& T<N\\
1/K,& T\geq N.
\end{array}
\right.
\end{equation}
Further generalizations of $T$-privacy, e.g., when privacy is required only against certain specified collusion patterns \cite{Tajeddine_Gnilke_Karpuk_Etal, Jia_Sun_Jafar} have also been explored. In particular, capacity is known for disjoint colluding sets \cite{Jia_Sun_Jafar}.

The rapidly growing body of literature in this area has  produced capacity results for PIR under a rich variety of constraints \cite{Banawan_Ulukus, FREIJ_HOLLANTI, Sun_Jafar_MDSTPIR, Lin_Kumar_Rosnes_Amat, Attia_Kumar_Tandon, Banawan_Ulukus_MPIR, Sun_Jafar_PIRL, Sun_Jafar_MPIR, Tian_Sun_Chen, Kadhe_Garcia_Heidarzadeh_Rouayheb_Sprintson, Chen_Wang_Jafar_Side, Tandon_CachePIR, Wei_Banawan_Ulukus, Wei_Banawan_Ulukus_Side, Shariatpanahi_Siavoshani_Maddah, Li_Gastpar}. However,  the capacity for the natural setting of secure storage remains unknown, and relatively unexplored. While a number of efforts  are motivated by security concerns, such efforts have  focused largely on other models, e.g., wiretap models  where data security is desired against eavesdroppers listening to the communication between the user and the servers \cite{Banawan_Ulukus_Asymmetric, Wang_Sun_Skoglund}, Byzantine  models where the servers may respond incorrectly by introducing erasures or errors in their response to the user's queries \cite{Banawan_Ulukus_Byzantine, Zhang_Ge_Variant, Tajeddine_Gnilke_Karpuk_Hollanti, wang2018e, yao2019capacity}, and so called symmetric security models \cite{Sun_Jafar_SPIR, Wang_Skoglund_TSPIR, Wang_Skoglund_SPIRAd} that allow the user to learn nothing about the data besides his desired message. An exception in this regard is the recent work in \cite{Yang_Shin_Lee} where PIR with distributed storage is explored and the asymptotic (large $K$) capacity for the $X=T=1$ setting is bounded as
\begin{align}
\left(1-\frac{1}{\sqrt{N}}\right)^2&\leq \lim_{K\rightarrow\infty}C_{\mbox{\tiny XSTPIR}}(N, K,  X=1, T=1)\leq \left(1-\frac{1}{N}\right).
\end{align}

As the main result of our work, we close this gap and characterize the asymptotic capacity of XSTPIR for all $N, X, T$ as follows.
\begin{align}
 \lim_{K\rightarrow\infty}C_{\mbox{\tiny XSTPIR}}(N, K,  X, T)
 &=\left\{
\begin{array}{ll}
1-\left(\frac{X+T}{N}\right),& N> X+T\\
0,& N\leq X+T.
\end{array}
\right.
\end{align}
The asymptotic capacity characterization leads us to supplementary results which include a general upper bound on the capacity of XSTPIR, the exact capacity characterization for any number of messages $K$ if $N\leq X+T$, and the exact capacity characterization for any $K$ if $X=T=1, N=3$. The results also lead us to interesting observations about aligning versus decoding noise, dependence of PIR rate on field size, robustness to  symmetric security constraints, and a particularly useful idea called cross subspace alignment. When privately retrieving multiple  symbols from a desired message in a secure distributed storage system, the structure (say, $1, \beta, \beta^2, \cdots$, for one symbol and $1, \gamma, \gamma^2, \cdots$ for another, as in Reed-Solomon (RS) codes) of  storage and queries for each symbol  determines the number of dimensions occupied by interference and the resolvability of desired symbols. Choosing  identical RS codes ($\beta=\gamma$) for each symbol of the same message would cause desired signals to align among themselves, while making the RS codes insufficiently dependent would cause interference to occupy too many dimensions. Cross subspace alignment is achieved by drawing the code parameters  as linear combinations from the same subspace (say, $\beta=1+\alpha, \gamma=2+\alpha$), which turns out to be the optimal way to align interference while keeping desired symbols resolvable. For a summary of results and a better explanation of the main  observations we refer the reader directly to  Section \ref{sec:resobs}. 

Let us start by defining our notation.

\emph{Notation: }Let $[m:n]$ denote the set $\{m,m+1,\dots,n\}$ for any two integers $m, n$ such that $m\leq n$. For sake of simplicity, let $X_{[m:n]}$ denote the set of random variables $\{X_m,X_{m+1},\dots,X_n\}$.  
For an index set $\mathcal{I}=\{i_1,i_2,\dots,i_n\}$, let $X_{\mathcal{I}}$ denote the set $\{X_{i_1},X_{i_2},\dots,X_{i_n}\}$. For variables $a_n, n\in[1:N]$ and an arbitrary function $f(\cdot)$, we denote the $N\times 1$ vector whose $n^{th}$ term is $f(a_n)$, as $\vect{f(a)}$. Similarly, $\vect{g(b)}$ denotes the vector $(g(b_1), \cdots, g(b_n))^T$ for variables $b_n, n \in [1:N]$ and a function $g(\cdot)$.
For such $N\times 1$ vectors $\vect{f(a)}$ and $\vect{g(b)}$, let $\vect{f(a)}\circ\vect{g(b)}$ denote their Hadamard product, i.e., the $N\times 1$ vector whose $n^{th}$ term is $f(a_n)\times g(b_n)$.
The notation $X\sim Y$ is used to indicate that $X$ and $Y$ are identically distributed. When a natural number, say $\ell\in\mathbb{N}$, is used to represent an element of a finite field $\mathbb{F}_q$, it denotes the sum of $\ell$ ones in $\mathbb{F}_q$, i.e., $\ell\triangleq \sum_{l=1}^\ell 1$, where the addition is over $\mathbb{F}_q$.

\section{XSTPIR: Problem Statement}
Consider data that is stored at $N$ distributed servers. The data consists of  $K$ independent messages, $W_1, W_2, \cdots, W_K$, and each message is represented\footnote{As usual for an information theoretic formulation, the actual size of each message is allowed to approach infinity. The parameters $L$ and $q$ partition the data into blocks and may be  chosen freely by the coding scheme to  match the code dimensions. Since the coding scheme for a block can be repeated for each successive block of data with no impact on rate, it suffices to consider one block of data subject to optimization over $L$ and $q$.} by $L$ random symbols from the finite field $\mathbb{F}_q$.
\begin{eqnarray}
H(W_1)=H(W_2)=\dots=H(W_K)=L, \label{h2}\\
H(W_1,W_2,\dots,W_K)=KL, \label{h1}
\end{eqnarray}
in $q$-ary units. 
There are $N$ servers. The information stored at the $n^{th}$ server is denoted by $S_n, n\in[1:N]$.  An $X$-secure scheme, $0\leq X <N$, guarantees that any $X$ (or fewer) colluding servers learn nothing about the data.
\begin{align}
\text{[$X$-Security] }&& I(S_{\mathcal{X}};W_1,\dots,W_K)&=0, &&\forall \mathcal{X}\subset [1:N], |\mathcal{X}|=X. \label{secur}
\end{align}
Besides $X$-security, we place no other constraint\footnote{The  amount of storage at each server is not constrained \emph{a priori}, however, it is remarkable that none of the XSTPIR schemes in this work end up storing  more than $KL$ symbols at each server. Thus the amount of storage used is not worse than a data replication scheme in the absence of security constraints.} on the amount of storage or the storage code used at each server, all of which is jointly optimized to maximize the capacity of XSTPIR. To ensure information retrieval is possible, note that the set of messages $W_1, \cdots, W_K$ must be a function of $S_{[1:N]}$.
\begin{align}\label{msgfunc}
H(W_1,\cdots, W_K\mid S_{[1:N]})&=0.
\end{align}
The user generates a desired message index $\theta$ privately and uniformly from $[1:K]$. In order to retrieve $W_\theta$ privately, the user generates $N$ queries, $Q_1^{[\theta]},Q_2^{[\theta]},\dots,Q_N^{[\theta]}$. The query  $Q_n^{[\theta]}$ is sent to the $n^{th}$ server.
The user has no prior knowledge of the information stored at the servers, i.e.,
\begin{equation}
I(S_{[1:N]}
;Q_{[1:N]}^{[\theta]}, \theta)=0.\label{indp}
\end{equation}
$T$-privacy, $1\leq T\leq N$, guarantees that any $T$ (or fewer) colluding servers learn nothing about $\theta$.
\begin{align}
\text{[$T$-Privacy] }&&I(Q_{\mathcal{T}}^{[\theta ]},  S_{\mathcal{T}};\theta)&=0, &&\forall \mathcal{T}\subset [1:N], |\mathcal{T}|=T. \label{privacy}
\end{align}
Upon receiving the query $Q_n^{[\theta]}$, the $n^{th}$ server generates an answering string $A_n^{[\theta]}$, as a function of the query $Q_n^{[\theta]}$ and its  stored information $S_n$. 
\begin{equation}
H(A_n^{[\theta]}|Q_n^{[\theta]},S_n)=0. \label{ans_det}
\end{equation}
From all the answers the user must be able to recover the desired message $W_\theta$,
\begin{align}
\text{[Correctness] }&& H(W_\theta|A_{[1:N]}^{[\theta]},Q_{[1:N]}^{[\theta]},\theta)=0. \label{corr}
\end{align}
The rate of an XSTPIR scheme characterizes how many bits of desired message are retrieved per downloaded bit, (equivalently, how many $q$-ary symbols of desired message are retrieved per downloaded $q$-ary symbol),
\begin{equation}
R= \frac{L}{D},
\end{equation}
where $D$ is the expected value (with respect to the random queries) of the number of $q$-ary symbols downloaded by the user from all servers.
The capacity of XSTPIR, denoted $C_{\mbox{\tiny XSTPIR}}(N, K, X, T)$, is  the supremum of achievable rates.

Finally, note that setting $X=0$ and $T=1$ reduces the XSTPIR problem to the basic PIR setting where data storage is not secure and the user's privacy is only guaranteed if no collusion takes place among servers. Setting $X=0$ for arbitrary $T$, reduces XSTPIR to the $T$-PIR problem. Setting $T=0$ for arbitrary $X$ reduces XSTPIR to an $X$-secure storage scheme with no privacy constraint.

\section{Capacity of XSTPIR: Results and Observations}\label{sec:resobs}

The results of this work are presented in this section, followed by  some observations.
\subsection{Results}\label{sec:results}
Our first result, presented in the following theorem, is an upper bound  on the capacity of XSTPIR.
\begin{theorem}\label{theorem:converse}
\begin{equation}
C_{\mbox{\tiny XSTPIR}}(N, K, X, T)\leq \left(\frac{N-X}{N}\right)C_{\mbox{\tiny TPIR}}(N-X, K, T).\label{eq:bound}
\end{equation}
\end{theorem}
The proof of Theorem \ref{theorem:converse} appears in Section \ref{proof:converse}. The intuition behind Theorem \ref{theorem:converse} may be understood through a thought experiment as follows. Without loss of generality, suppose the expected number of bits downloaded from each server is the same. Now, relax the constraints so that $S_{[1:X]}$, i.e., the stored information at the first $X$ servers is made available globally (to all servers and to the user) for free, the messages $W_1, W_2, \cdots, W_K$ are made available to all servers, and the data-security constraint is  eliminated. None of  this can hurt capacity because any XSTPIR scheme from before can still be used with the relaxed constraints. So any upper bound on capacity of this relaxed setting is still an upper bound on the capacity of the original XSTPIR setting. The relaxed setting is analogous to the  $T$-PIR problem with $K$ messages and $N-X$ servers, for which we already know the optimal download per server from the existing capacity results for $T$-PIR. 
Thus, the statement of Theorem \ref{theorem:converse} follows. However, formalizing this intuition into a proof is not trivial because of the correlated side-information generated at the user and servers in the process of relaxing the constraints. Indeed, the formal proof presented in Section \ref{proof:converse} takes a less direct approach.

It turns out the bound in Theorem \ref{theorem:converse}  is quite powerful. In fact, we suspect that this bound might be tight in general. An immediate observation is that if we set $X=0$, i.e., remove the data storage security constraint, then the bound is tight because it gives us the capacity of $T$-PIR. Similarly, if we set $T=0$, i.e., the privacy constraint is removed, then the bound is also tight, and the capacity in the absence of privacy constraints is easily seen to be $C_{\mbox{\tiny XSTPIR}}(N, K, X, T=0) = 1-\frac{X}{N}$, which is achievable by a simple secret-sharing scheme. We further prove the tightness of this bound for the   cases identified in our next set of results. The first setting identifies a somewhat degenerate extreme where  it is optimal to download everything. 
\begin{theorem}\label{theorem:mds}
 If $N\leq X+T$, then\footnote{Note that $N > X$ by definition.} for arbitrary $K$,
\begin{align}
C_{\mbox{\tiny XSTPIR}}(N, K, X, T)&= \left(\frac{N-X}{N}\right)C_{\mbox{\tiny TPIR}}(N-X, K, T)\\
&=\frac{N-X}{NK}.
\end{align}
\end{theorem}
The proof of Theorem \ref{theorem:mds} is presented in Section \ref{proof:mds}. Since the upper bound is already provided by Theorem \ref{theorem:converse}, only a proof of achievability is needed. Furthermore, since retrieving the desired message  in this setting amounts to downloading everything stored at all servers regardless of which message is desired, the only thing required for the achievable scheme is a secure storage scheme, which is readily achieved by including $X$ uniformly random noise symbols for every $N-X$ symbols of each message.

Next, the main result of this paper is the asymptotic capacity characterization presented in the following theorem. 
\begin{theorem} \label{theorem:main} As the number of messages $K\rightarrow\infty$,  for arbitrary $N, X, T$, 
\begin{align}
\lim_{K\rightarrow\infty}C_{\mbox{\tiny XSTPIR}}(N, K, X, T)&= \lim_{K\rightarrow\infty}\left(\frac{N-X}{N}\right)C_{\mbox{\tiny TPIR}}(N-X, K, T)\\
&=
\left\{
\begin{array}{ll}
1-\left(\frac{X+T}{N}\right),& N> X+T\\
0,& N\leq X+T.
\end{array}
\right.
\end{align}
\end{theorem}
The proof of Theorem \ref{theorem:main} appears in Section \ref{proof:main}. Theorem \ref{theorem:main} is   significant  for two reasons. First, asymptotic capacity results are particularly relevant for PIR problems because the capacity approaches its asymptotic value extremely quickly ---  the gap is negligible even for moderate values of $K$, and $K$ is typically a large value. Second, the asymptotic capacity result showcases a new idea, \emph{cross subspace alignment}, that is interesting by itself.

\bigskip

Insights from the asymptotically optimal scheme allow us to settle the exact capacity of  XSTPIR  with $X=T=1$, $N=3$ and arbitrary $K$.
\begin{theorem}\label{theorem:N3}
 If the number of servers, $N=3$, and $X=T=1$, then for arbitrary number of messages, $K$,
\begin{align}
C_{\mbox{\tiny XSTPIR}}(N=3, K, X=1, T=1)&= \left(\frac{N-X}{N}\right)C_{\mbox{\tiny TPIR}}(N-X, K, T)\\
&=\frac{2}{3}\left(1+\frac{1}{2}+\frac{1}{2^2}+\cdots+\frac{1}{2^{K-1}}\right)^{-1}.
\end{align}
\end{theorem}
Theorem \ref{theorem:N3} is proved in Section \ref{proof:N3}. The capacity achieving scheme introduces a new insight. For almost all  PIR settings studied so far, asymptotic capacity achieving schemes have been found that send a uniformly random query vector to each server and download a product of the query vector and  information stored at the server.  Suppose the query vector is uniform over $\mathbb{F}_q^M$. Then with probability $1/q^M$ the query vector is all zero, and the scheme requests nothing from the server. Typically $M$ depends on the number of messages $K$.  As $K$ approaches infinity the probability of requesting nothing approaches zero, so this does not help in the asymptotic sense. However, if the same scheme is used for finite $K$, then $M$ is also finite, $1/q^M>0$, and the average download is reduced by the factor $(1-1/q^M)$, which improves the achieved rate of the scheme. It is remarkable that the rate achieved in this way depends on the field size. This idea is essential to the capacity achieving scheme for Theorem \ref{theorem:N3}. 

Next we present some observations that place our results in perspective.
 
\subsection{Observations}\label{sec:insights}
\subsubsection{Alignment of Noise and Interference}
\begin{figure}[h]
\centering
\includegraphics[width=8cm]{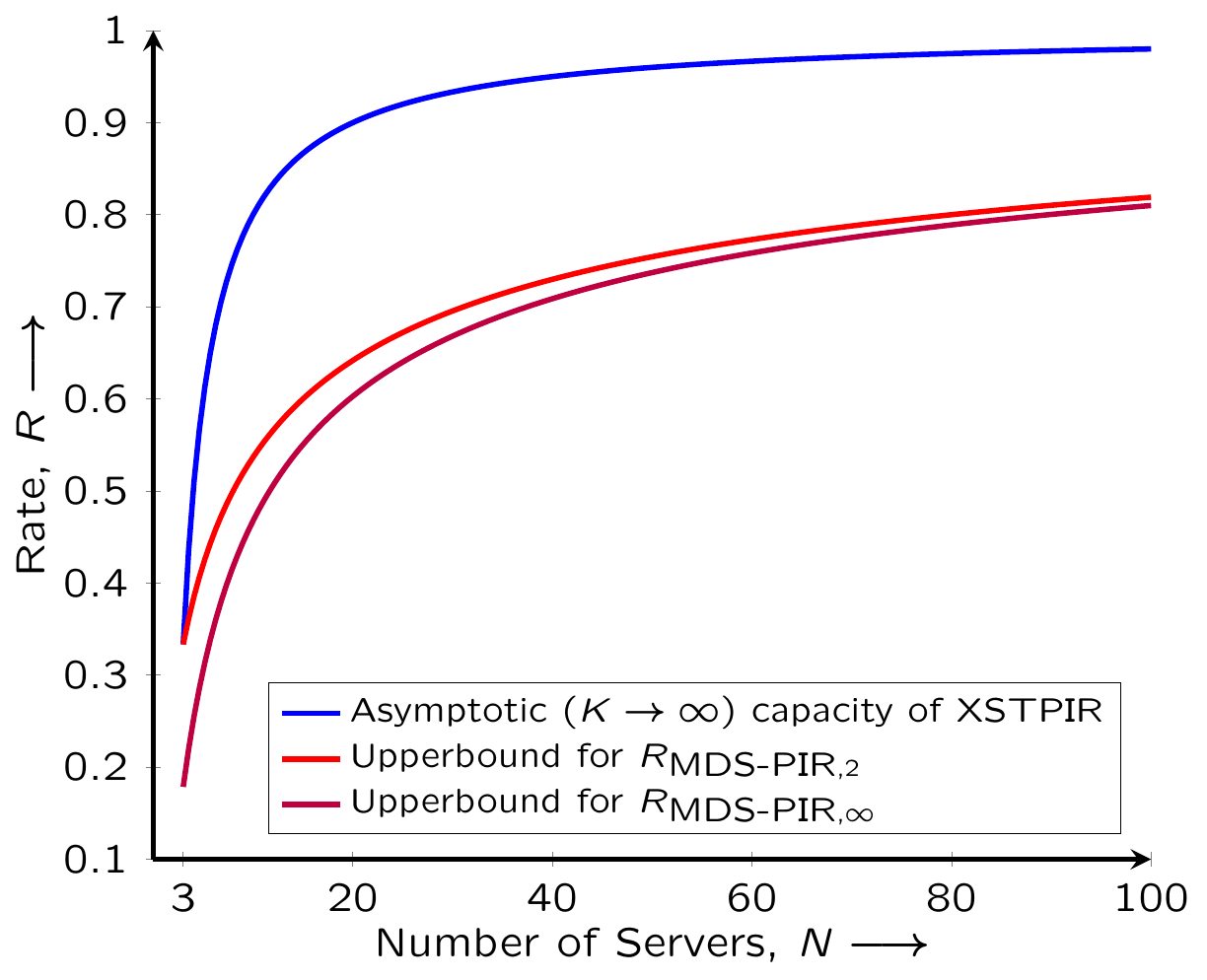}
\caption{\small \it Suboptimality of the rate achieved by the $X=1$ secure MDS-PIR alternative that allows the user to decode noise relative to the rate achieved with the asymptotically optimal XSTPIR scheme where the noise is aligned with other interference.}
\label{comp}
\end{figure}

Consider the simplest non-trivial setting for XSTPIR, where $X=1, T=1$, and the number of servers, $N\geq 3$. 
A natural idea for providing $X=1$ secure storage is to include $1$ independent uniformly random noise symbol along with the $L$ symbols of each message, creating a new message with $M=L+1$ symbols. This new message is stored  across $N$ servers according to an $(N, M)$ MDS code, essentially storing a linear combination of the $M$ message symbols at each server, where the coefficients for the noise symbol at each server must be non-zero. Capacity is known for PIR with coded storage (MDS-PIR \cite{Banawan_Ulukus}), and one might wonder if such an MDS-PIR scheme might suffice to achieve capacity with secure storage. It is not difficult to see that the best rate achievable with such an MDS-PIR scheme is
\begin{eqnarray}
R_{\text{MDS-PIR}}=\frac{M-1}{M}\left(1+\left(\frac{M}{N}\right)+\dots+\left(\frac{M}{N}\right)^{K-1}\right)^{-1}.
\end{eqnarray}
The  $\frac{M-1}{M}$ penalty appears because one of the $M$ symbols of the \emph{decoded} message is the noise symbol. As $K\rightarrow \infty$, the rate approaches $R_{\text{MDS-PIR},\infty}=\frac{M-1}{M}\left(1-\left(\frac{M}{N}\right)\right)$. This expression takes its maximum value when $M=\sqrt{N}$, so it can be bounded as,
\begin{eqnarray}
R_{\text{MDS-PIR},\infty}&\leq&\frac{\sqrt{N}-1}{\sqrt{N}}\left(1-\left(\frac{\sqrt{N}}{N}\right)\right)=\left(1-\frac{1}{\sqrt{N}}\right)^2.
\end{eqnarray}
Note that this expression matches the achievable rate bound of \cite{Yang_Shin_Lee}. However, it is strictly smaller than, $1-2/N$, the asymptotic capacity of XSTPIR for this setting.  Evidently, the natural MDS-PIR solution, and the secret sharing based scheme of \cite{Yang_Shin_Lee},   are asymptotically suboptimal. In fact, the MDS-PIR solution falls short of the asymptotic ($K\rightarrow\infty$) capacity of XSTPIR, even if the MDS-PIR scheme is only required to deal with $K=2$ messages. Denoting the corresponding rate of the MDS-PIR scheme as $R_{\text{MDS-PIR},2}$, we have, 
\begin{align}
R_{\text{MDS-PIR},2}
&\leq\frac{\sqrt{N+1}}{\sqrt{N+1}+1}\left(1+\left(\frac{\sqrt{N+1}+1}{N}\right)\right)^{-1} \leq 1-\frac{2}{N}.
\end{align}

Figure \ref{comp} shows that the gap between the $X=1$ secure MDS-PIR alternative and the XSTPIR scheme is significant. Intuitively, the reason for this gap is the following. The secure MDS-PIR alternative allows the user to \emph{decode} the artificial noise symbol which is added to the message to guarantee security. However, in the XSTPIR scheme, the user is able to decode only  the desired message, and not the noise protecting it. In fact this noise is \emph{aligned} with other interfering symbols, e.g., the noise terms protecting other message symbols, thus creating a more efficient solution. Incidentally, the alignment of noise provides another unexpected benefit, in some cases it automatically makes the scheme symmetrically secure, as explained next.

\subsubsection{Symmetric Security: Capacity of Sym-XSPIR}
Let us fix $T=1$, thereby relaxing the $T$-privacy constraint to its minimum value for PIR. Now, suppose in addition to $X$-secure storage, we also include the so called `symmetric' security constraint, that the user should learn nothing about the data besides his desired message, i.e., 
\begin{align}
\text{[Sym-Security] }&& I(W_{[1:K]}; A_{[1:N]}^{[\theta]}\mid Q_{[1:N]}^{[\theta]}, W_\theta, \theta)&=0.
\end{align}
Capacity of the basic $(X=0, T=1, K>1)$ Sym-PIR setting  was shown in \cite{Sun_Jafar_SPIR} to be 
\begin{align}
C_{\mbox{\tiny Sym-PIR}}(K,N)=1-\frac{1}{N}.
\end{align}
 Note that there is a loss of capacity due to the additional symmetric security constraint. Furthermore,  the capacity without the symmetric security constraint depends on the number of messages $K$ while the capacity with the symmetric security constraint does not. 

XSTPIR with the symmetric security constraint and with $T=1$, in short the Sym-XSPIR setting (note that we drop the $T$ because $T=1$ is the degenerate case for $T$-privacy), reveals a surprising aspect of our XSTPIR schemes, that imposing the symmetric security constraint does not affect\footnote{For $T>1$ our XSTPIR schemes are not symmetrically secure.} our capacity results for $T=1$. This is made explicit in the following corollaries for Sym-XSPIR, that  match  the corresponding theorems for XSTPIR.

\begin{corollary}\label{cor:converse}
\begin{equation}
C_{\mbox{\tiny Sym-XSPIR}}(N, K, X)\leq \left(\frac{N-X}{N}\right)C_{\mbox{\tiny PIR}}(N-X, K).
\end{equation}
\end{corollary}

\begin{corollary}\label{cor:mds}
 If  $N=X+1$, then\footnote{Note that since $X< N$ by definition, and $T=1$ for XSPIR,  the condition $N\leq X+T$ is equivalent to $N=X+1$.} for arbitrary $K$,
\begin{align}
C_{\mbox{\tiny Sym-XSPIR}}(N, K, X)&= \left(\frac{N-X}{N}\right)C_{\mbox{\tiny PIR}}(N-X, K)\\
&=\frac{1}{NK}.
\end{align}
\end{corollary}

\begin{corollary} \label{cor:main} As the number of messages $K\rightarrow\infty$,  for arbitrary $N, X$, 
\begin{align}
\lim_{K\rightarrow\infty}C_{\mbox{\tiny Sym-XSPIR}}(N, K, X)&= \lim_{K\rightarrow\infty}\left(\frac{N-X}{N}\right)C_{\mbox{\tiny PIR}}(N-X, K)\\
&=
\left\{
\begin{array}{ll}
1-\left(\frac{X+1}{N}\right),& N> X+1\\
0,& N\leq X+1.
\end{array}
\right.
\end{align}
\end{corollary}

\begin{corollary}\label{cor:N3}
 If the number of servers, $N=3$, and $X=1$, then for arbitrary number of messages, $K$,
\begin{align}
C_{\mbox{\tiny Sym-XSPIR}}(N=3, K, X=1)&= \left(\frac{N-X}{N}\right)C_{\mbox{\tiny PIR}}(N-X, K)\\
&=\frac{2}{3}\left(1+\frac{1}{2}+\frac{1}{2^2}+\cdots+\frac{1}{2^{K-1}}\right)^{-1}.
\end{align}
\end{corollary}
The proofs of all $4$ corollaries appear in Appendix \ref{proof:cor}. Surprisingly, note that there is no loss of capacity in each case due to the additional symmetric security constraint.  Also note that according to Corollary \ref{cor:N3}, unlike Sym-PIR, the capacity of Sym-XSPIR depends on the number of messages $K$  for all $K>1$.

\subsubsection{Cross Subspace Alignment}\label{sec:cross}
Conceptually, the most intriguing aspect of the asymptotically optimal XSTPIR scheme is the extent to which it is able to align interference. Interference alignment is central to PIR \cite{Sun_Jafar_BIAPIR, Sun_Jafar_PIR}, and nearly all existing PIR constructions use some form of interference alignment. The strength of XSTPIR lies in the novel idea of cross subspace alignment, that we explain intuitively in this section through an example. 
Consider the setting of $X=2$ secure and $T=1$ PIR with $N=5$ servers. Let $w_{1}$ be a symbol from a desired message $W$. For simplicity (and because identical alignments are applied to all messages), it suffices to focus on only this message for the purpose of this explanation. In order to guarantee $X=2$ security, $w_1$ is {\color{black}mixed} with $2$ random noise symbols $z_{11}, z_{12}$,  according to the following RS Code, so that the $n^{th}$ row is stored at the $n^{th}$ server, $n\in[1:5]$.

\begin{align*}
&
\left[\begin{matrix}
1\\
1\\
1\\
1\\
1\\
\end{matrix}
\right]w_{1}+\left[\begin{matrix}
\beta_1\\
\beta_2\\
\beta_3\\
\beta_4\\
\beta_5
\end{matrix}
\right]z_{11}
+\left[\begin{matrix}
\beta_1^2\\
\beta_2^2\\
\beta_3^2\\
\beta_4^2\\
\beta_5^2
\end{matrix}
\right]z_{12}\\
&\triangleq\vect{1}w_{1}+\vect{\beta}z_{11}+\vect{\beta^{2}}z_{12} .
\end{align*}
To ensure privacy, the query symbol $q_\theta$ ($q_\theta=1$, i.e., this message is desired) is similarly mixed with a noise symbol $z_{1}'$.
\begin{align*}
\vect{1}q_\theta+\vect{\beta}z_1'=\vect{1}+\vect{\beta}z_1'
\end{align*}
and the $n^{th}$ row of this query vector is sent to the $n^{th}$ server. Each server returns the product of the noisy query symbol and the noisy stored symbol, so that the user receives the $5$ answers.
\begin{align*}
&\left(\vect{1}w_1+\vect{\beta}z_{11}+\vect{\beta^{2}}z_{12}\right)\circ\left(\vect{ 1}+\vect{\beta}z_1'\right)\\
&=\vect{1}w_1+\vect{\beta}\left(w_1z_1'+z_{11}\right)+\vect{\beta^{2}}\left(z_{11}z_1'+z_{12}\right)+\vect{\beta^{3}}z_{12}z_1' .
\end{align*}
The desired symbol $w_1$ appears along the vector $\vect{1}$ while the remaining $5$ undesired symbols align along $3$ dimensions. Specifically, the undesired symbols $w_1z_1'$ and $z_{11}$  align along the vector $\vect{\beta}$; undesired symbols $z_{11}z_1'$ and $z_{12}$  align along  the vector $\vect{\beta^{2}}$ and  undesired symbol $z_{12}z_1'$ appears along the vector $\vect{\beta^{3}}$. This type of alignment, enabled by using the same $\vect{\beta}$ in the storage and query, is indeed very useful and has been used previously by Freij-Hollanti et al. for MDS-TPIR \cite{FREIJ_HOLLANTI}. However, note that we have a $5$ dimensional space (all vectors are $5\times 1$) and we are so far only using  $4$ dimensions (one desired, three interference), so there is room for improvement. 

In order to improve the efficiency of the retrieval scheme, suppose we try to retrieve another symbol, $w_2$, from the same desired message $W=(w_1, w_2)$. The challenge is that because of the $X=2$ security requirement $w_2$ is mixed with new (independent) noise symbols $z_{21}, z_{22}$ according to an RS code parameterized by $\gamma$, 
\begin{align}
\vect{1}w_{2}+\vect{\gamma}z_{21}+\vect{\gamma^{2}}z_{22},
\end{align}
so any attempt to retrieve $w_2$ will add new interference terms. Since we already have $3$ dimensions of interference, the new interference added due to the noise protecting $w_2$ must align completely within the existing interference. This will be accomplished by cross-alignment, i.e., introducing additional structure across the storage and query codes for the different symbols to be retrieved. 
In particular, we will use the query vector $\vect{\gamma}\circ \left(\vect{1}+\vect{\beta}z_1'\right)$ to multiply with the stored variables containing $w_1$ (i.e., $\vect{1}w_{1}+\vect{\beta}z_{11}+\vect{\beta^{2}}z_{12}$) and the query vector $\vect{\beta}\circ \left(\vect{1}+\vect{\gamma}z_2'\right)$ to multiply with the stored variables containing $w_2$ (i.e., $\vect{1}w_{2}+\vect{\gamma}z_{21}+\vect{\gamma^{2}}z_{22}$). The sum of the two multiplications is returned as the answer. Note that Hadamard products are commutative and associative.
The answers from the $5$ servers are now expressed  as follows.
\begin{align}
&\vect{\gamma}\circ\left(\vect{1}w_1+\vect{\beta}z_{11}+\vect{\beta^{2}}z_{12}\right)\circ\left(\vect{1}+\vect{\beta}z_1'\right)+\vect{\beta}\circ(\vect{1}w_2+\vect{\gamma}z_{21}+\vect{\gamma^{2}}z_{22})\circ\left(\vect{1}+\vect{\gamma}z_2'\right)\\
&=\vect{\gamma}w_1+\vect{\beta}w_2+\vect{\beta}\circ\vect{\gamma}\left(w_1z_1'+z_{11}+w_2z_2'+z_{21}\right)\nonumber\\
&\hspace{1cm}+ \vect{\beta^{2}}\circ\vect{\gamma}(z_{11}z_1'+z_{12})+\vect{\beta}\circ\vect{\gamma^{2}}(z_{21}z_2'+z_{22})+\vect{\beta^{3}}\circ\vect{\gamma}z_{12}z_1'+\vect{\beta}\circ\vect{\gamma^{3}}z_{22}z_2' .
\end{align}
Note that we cannot choose $\vect{\beta} = \vect{\gamma}$, because the two desired symbols $(w_1, w_2)$ must not align in the same dimension. Also note that by cross-multiplying the first set of answers with $\vect{\gamma}$ and the second with $\vect{\beta}$ we have achieved  \emph{cross alignment} of $4$ terms along $\vect{\beta} \circ \vect{\gamma}$. However, we now have $5$  dimensions occupied by interference, along the $5$ vectors, $\vect{\beta}\circ\vect{\gamma}, \vect{\beta^{2}}\circ\vect{\gamma}, \vect{\beta}\circ\vect{\gamma^{2}}, \vect{\beta^{3}}\circ\vect{\gamma}, \vect{\beta}\circ\vect{\gamma^{3}}$. Since the overall space is only $5$ dimensional and we need two dimensions for desired symbols, we need to restrict interference to no more than $3$ dimensions. Surprisingly, it is possible to do this by cross \emph{subspace} alignment as we show next. Let us introduce a  structural relationship between $\beta$ and $\gamma$. In particular,  let us set,
\begin{align}
\vect{\beta}&=\vect{1+\alpha}\\
\vect{\gamma}&=\vect{2+\alpha}
\end{align}
so that the answers from the $5$ servers are now expressed as,
\begin{align}
\left(\vect{2+\alpha}\right)w_1+
\left(\vect{1+\alpha}\right)w_2+\left(\vect{1+\alpha}\right)\circ\left(\vect{2+\alpha}\right)I
\end{align}
where the interference $I$ is 
\begin{align}
I&=\vect{1}(w_1z_1'+z_{11}+w_2z_2'+z_{21})+\left(\vect{1+\alpha}\right)(z_{11}z_1'+z_{12})+\left(\vect{2+\alpha}\right)(z_{21}z_1'+z_{22})\nonumber\\
&+\left(\vect{1 + 2\alpha+\alpha^{2}}\right)z_{12}z_1'+\left(\vect{4+4\alpha+\alpha^{2}}\right)z_{22}z_2' .
\end{align}
Note that there are still $5$ interference vectors, no two of which align directly with each other. However, the $5$ interference vectors align into a $3$ dimensional subspace of the $5$ dimensional vector space. This is what we mean by  \emph{cross subspace alignment} and it is essential to  this work. To see explicitly how the interference aligns into a $3$ dimensional subspace, we can rewrite $I$ as,
\begin{align}
I&=\vect{1}(w_1z_1'+z_{11}+w_2z_2'+z_{21}+z_{11}z_1'+z_{12}+2z_{21}z_1'+2z_{22}+z_{12}z_1'+4z_{22}z_2')\nonumber\\
&+\vect{\alpha}(z_{11}z_1'+z_{12}+z_{21}z_1'+z_{22}+2z_{12}z_1'+4z_{22}z_2')\nonumber\\
&+\vect{\alpha^{2}}(z_{12}z_1'+z_{22}z_2') .
\end{align}
Thus, due to cross subspace alignment, all of $I$ aligns within a $3$ dimensional space, leaving the remaining $2$ dimensions interference-free for the desired symbols.  Exactly the same alignments apply to all messages as explained in the formal descriptions of the schemes provided in this paper.

\section{Proof of Theorem \ref{theorem:converse}}\label{proof:converse}
\allowdisplaybreaks
Let us start with two useful lemmas. The first one shows that the desired message index is independent of the messages, stored variables, queries and answers.

\begin{lemma}\label{lemma:sprivacy}
For all $k, k' \in [1:K], \forall \mathcal{T} \in [1:N], |\mathcal{T}| = T$, we have
\begin{align}
(Q_{\mathcal{T}}^{[k]}, A_{\mathcal{T}}^{[k]},  S_{[1:N]}, W_1,\cdots, W_K) \sim  (Q_{\mathcal{T}}^{[k']}, A_{\mathcal{T}}^{[k']},  S_{[1:N]}, W_1,\cdots, W_K)
\label{eq:sprivacy}
\end{align}
\end{lemma}

{\it Proof:} Since $W_1,\cdots, W_K$ is a function of $S_{[1:N]}$ and $A_{\mathcal{T}}^{[\theta]}$ is a function of $(Q_{\mathcal{T}}^{[\theta]}, S_{\mathcal{T}})$ (refer to (\ref{ans_det})), it suffices to prove $I(\theta; Q_{\mathcal{T}}^{[\theta ]},  S_{[1:N]}) = 0$. From (\ref{indp}), we have
\begin{align}
&I(Q_{[1:N]}^{[\theta]}, \theta; S_{[1:N]}) = 0 \\
\Rightarrow~~& I(Q_{\mathcal{T}}^{[\theta]}, \theta; S_{[1:N]}) = 0\\
\Rightarrow~~& I(Q_{\mathcal{T}}^{[\theta]} ; S_{[1:N]}) = I(Q_{\mathcal{T}}^{[\theta]} ; S_{[1:N]} | \theta)  = 0 \label{eq:ss}
\end{align}
Next, we have,
\begin{align}
 I(\theta; Q_{\mathcal{T}}^{[\theta ]},  S_{[1:N]}) 
 &\overset{(\ref{indp})}{=} I(\theta; Q_{\mathcal{T}}^{[\theta ]} | S_{[1:N]}) \\
&= H(Q_{\mathcal{T}}^{[\theta ]} | S_{[1:N]})  - H(Q_{\mathcal{T}}^{[\theta ]} | S_{[1:N]}, \theta) \\
&\overset{(\ref{eq:ss})}{=} H(Q_{\mathcal{T}}^{[\theta ]}) - H(Q_{\mathcal{T}}^{[\theta ]}|\theta) \\
 &\overset{(\ref{privacy})}{=} 0
\end{align}
$\hfill\square$

The second lemma is a statement of conditional independence of answers from one set of servers from the queries to the rest of the servers.
\begin{lemma}
For all $\mathcal{T}, \mathcal{X} \subset [1:N], \forall k \in [1:K], \forall \mathcal{K} \in [1:K]$, we have
\begin{align}
 H(A_{\mathcal{T}}^{[k]} | S_{\mathcal{X}}, Q_{[1:N]}^{[k]}, W_{\mathcal{K}}) = H(A_{\mathcal{T}}^{[k]} | S_{\mathcal{X}}, Q_{\mathcal{T}}^{[k]}, W_{\mathcal{K}}) \label{eq:mar}
\end{align}
\end{lemma}

{\it Proof:} It suffices to prove that $I(A_{\mathcal{T}}^{[k]}; Q_{[1:N]}^{[k]} | S_{\mathcal{X}}, Q_{\mathcal{T}}^{[k]}, W_{\mathcal{K}}) = 0$. This proof is presented as follows.
\begin{align}
I(A_{\mathcal{T}}^{[k]}; Q_{[1:N]}^{[k]} | S_{\mathcal{X}}, Q_{\mathcal{T}}^{[k]}, W_{\mathcal{K}}) &\leq I(A_{\mathcal{T}}^{[k]}, S_{\mathcal{X}}, W_{\mathcal{K}}; Q_{[1:N]}^{[k]} | Q_{\mathcal{T}}^{[k]}) \\
&\leq I(A_{\mathcal{T}}^{[k]}, S_{[1:N]}, W_{\mathcal{K}}; Q_{[1:N]}^{[k]} | Q_{\mathcal{T}}^{[k]}) \\
&\overset{(\ref{msgfunc})(\ref{ans_det})}{=} I(S_{[1:N]}; Q_{[1:N]}^{[k]} | Q_{\mathcal{T}}^{[k]}) \\
&\overset{(\ref{indp})}{=} 0
\end{align}

$\hfill\square$

The next lemma formalizes the intuition that because of the security constraint, the answers from any $X$ servers are, in some sense, not very useful. Specifically, after conditioning on the information contained in any $X$ servers,  the answers from the remaining $N-X$ servers must still contain at least $L$ more bits than the interference that is included in those answers. For a set $\mathcal{X}$, its complement set is denoted as $\overline{\mathcal{X}}$, i.e., $\overline{\mathcal{X}} = \{n | n \in [1:N], n \notin \mathcal{X}\}$. We use $D_n$ to denote the expected number of symbols downloaded from Server $n$. 
\begin{lemma}\label{lemma:int}
For all $ \mathcal{X} \in [1:N], |\mathcal{X}| = X$, we have
\begin{align}
L \leq \sum_{n \in \overline{\mathcal{X}}}D_n - 
H(A_{\overline{\mathcal{X}}}^{[1]} | S_{\mathcal{X}}, Q_{[1:N]}^{[1]}, W_1)
\label{eq:int}
\end{align}
\end{lemma}
{\it Proof:}
\begin{align}
L = H(W_1) &\overset{(\ref{corr})}{=} I(W_1; A_{[1:N]}^{[1]} | Q_{[1:N]}^{[1]}) \\
&\leq I(W_1; A_{[1:N]}^{[1]}, S_{\mathcal{X}} | Q_{[1:N]}^{[1]}) \\
&=  I(W_1; S_{\mathcal{X}} | Q_{[1:N]}^{[1]}) +  I(W_1; A_{\mathcal{X}}^{[1]}, A_{\overline{\mathcal{X}}}^{[1]} | S_{\mathcal{X}}, Q_{[1:N]}^{[1]}) \\
&\overset{(\ref{ans_det})}{=}  I(W_1; S_{\mathcal{X}} | Q_{[1:N]}^{[1]}) +  I(W_1; A_{\overline{\mathcal{X}}}^{[1]} | S_{\mathcal{X}}, Q_{[1:N]}^{[1]}) \\
&\overset{(\ref{indp})}{=} I(W_1, Q_{[1:N]}^{[1]}; S_{\mathcal{X}}) +  I(W_1; A_{\overline{\mathcal{X}}}^{[1]} | S_{\mathcal{X}}, Q_{[1:N]}^{[1]}) \\
&\overset{(\ref{secur})}{=} I(Q_{[1:N]}^{[1]}; S_{\mathcal{X}}| W_1) +  I(W_1; A_{\overline{\mathcal{X}}}^{[1]} | S_{\mathcal{X}}, Q_{[1:N]}^{[1]}) \\
&\overset{}{\leq} I(Q_{[1:N]}^{[1]}; S_{\mathcal{X}},W_1) +  I(W_1; A_{\overline{\mathcal{X}}}^{[1]} | S_{\mathcal{X}}, Q_{[1:N]}^{[1]}) \\
&\overset{(\ref{indp})}{=}  I(W_1; A_{\overline{\mathcal{X}}}^{[1]} | S_{\mathcal{X}}, Q_{[1:N]}^{[1]}) \label{eq:red} \\
&\overset{}{\leq} \sum_{n \in \overline{\mathcal{X}}}D_n - H(A_{\overline{\mathcal{X}}}^{[1]} | S_{\mathcal{X}}, Q_{[1:N]}^{[1]}, W_1)
\end{align}
$\hfill\square$

We may interpret the second term of the RHS of (\ref{eq:int}) as the interference term. To bound it, we need the following recursive relation, stated in a lemma.

\begin{lemma}\label{lemma:recursive}
For all $\mathcal{X} \in [1:N], |\mathcal{X}| = X$ and for all $k \in [1:K]$, we have
\begin{align}
& H(A_{\overline{\mathcal{X}}}^{[k]} | S_{\mathcal{X}}, Q_{[1:N]}^{[k]}, W_{[1:k]}) \geq \frac{T}{N-X} \left( L + H(A_{\overline{\mathcal{X}}}^{[k+1]} | S_{\mathcal{X}}, Q_{[1:N]}^{[{k+1}]}, W_{[1:k+1]}) \right), ~\mbox{if}~ N > X+T.\\
& H(A_{\overline{\mathcal{X}}}^{[k]} | S_{\mathcal{X}}, Q_{[1:N]}^{[k]}, W_{[1:k]}) \geq L + H(A_{\overline{\mathcal{X}}}^{[k+1]} | S_{\mathcal{X}}, Q_{[1:N]}^{[{k+1}]}, W_{[1:k+1]}), ~\mbox{if}~ N\leq X+T.
\label{eq:recursive}
\end{align}
\end{lemma}
{\it Proof:} First consider $N> X+T$. Consider any set $\mathcal{T} \subset \overline{\mathcal{X}}, |\mathcal{T}| = T$.
\begin{align}
 H(A_{\overline{\mathcal{X}}}^{[k]} | S_{\mathcal{X}}, Q_{[1:N]}^{[k]}, W_{[1:k]}) &\geq H(A_{\mathcal{T}}^{[k]} | S_{\mathcal{X}}, Q_{[1:N]}^{[k]}, W_{[1:k]}) \\
 &\overset{(\ref{eq:mar})}{=} H(A_{\mathcal{T}}^{[k]} | S_{\mathcal{X}}, Q_{\mathcal{T}}^{[k]}, W_{[1:k]})\\
 &\overset{(\ref{eq:sprivacy})}{=} H(A_{\mathcal{T}}^{[k+1]} | S_{\mathcal{X}}, Q_{\mathcal{T}}^{[k+1]}, W_{[1:k]}) \\
 &\overset{(\ref{eq:mar})}{=} H(A_{\mathcal{T}}^{[k+1]} | S_{\mathcal{X}}, Q_{[1:N]}^{[k+1]}, W_{[1:k]})  \label{eq:hh1}
\end{align}
Averaging (\ref{eq:hh1}) over all choices of $\mathcal{T}$ and applying Han's inequality, we have
\begin{align}
&~~~~~~~ H(A_{\overline{\mathcal{X}}}^{[k]} | S_{\mathcal{X}}, Q_{[1:N]}^{[k]}, W_{[1:k]}) \notag\\
&\geq \frac{T}{N-X} H(A_{\overline{\mathcal{X}}}^{[k+1]} | S_{\mathcal{X}}, Q_{[1:N]}^{[k+1]}, W_{[1:k]}) \\
&\overset{(\ref{ans_det})(\ref{corr})}{=}  \frac{T}{N-X} H(A_{\overline{\mathcal{X}}}^{[k+1]}, W_{k+1} | S_{\mathcal{X}}, Q_{[1:N]}^{[k+1]}, W_{[1:k]}) \\
&\overset{}{=}  \frac{T}{N-X} \left( H(W_{k+1} | S_{\mathcal{X}}, Q_{[1:N]}^{[k+1]}, W_{[1:k]}) + H(A_{\overline{\mathcal{X}}}^{[k+1]} | S_{\mathcal{X}}, Q_{[1:N]}^{[k+1]}, W_{[1:k+1]})  \right)\\
&\overset{}{=}  \frac{T}{N-X} \left( L + H(A_{\overline{\mathcal{X}}}^{[k+1]} | S_{\mathcal{X}}, Q_{[1:N]}^{[k+1]}, W_{[1:k+1]})  \right)
\end{align}
where the last step uses $L = H(W_{k+1})$ and $I(W_{k+1} ; S_{\mathcal{X}}, Q_{[1:N]}^{[k+1]}, W_{[1:k]}) = 0$, proved as follows.
\begin{align}
I(W_{k+1} ; S_{\mathcal{X}}, Q_{[1:N]}^{[k+1]}, W_{[1:k]}) &\overset{(\ref{h2})(\ref{h1})}{=} I(W_{k+1} ;  S_{\mathcal{X}}, Q_{[1:N]}^{[k+1]} \mid W_{[1:k]}) \\
&\leq I(W_{[1:k+1]} ;  S_{\mathcal{X}}, Q_{[1:N]}^{[k+1]}) \\
&\overset{(\ref{secur})}{=}  I(W_{[1:k+1]} ; Q_{[1:N]}^{[k+1]} |  S_{\mathcal{X}}) \\
&\leq  I(W_{[1:k+1]}, S_{\mathcal{X}} ; Q_{[1:N]}^{[k+1]} ) \\
&\leq  I(S_{[1:N]} ; Q_{[1:N]}^{[k+1]} ) \\
&\overset{(\ref{indp})}{=} 0 \label{eq:ind1}
\end{align}

Next, consider $N \leq X+T$. The proof is similar to that presented above. Note that $|\mathcal{X}| = N-X \leq T$.
\begin{align}
 H(A_{\overline{\mathcal{X}}}^{[k]} | S_{\mathcal{X}}, Q_{[1:N]}^{[k]}, W_{[1:k]}) &\overset{(\ref{eq:mar})}{=} H(A_{\overline{\mathcal{X}}}^{[k]} | S_{\mathcal{X}}, Q_{\overline{\mathcal{X}}}^{[k]}, W_{[1:k]})\\
 &\overset{(\ref{eq:sprivacy})}{=} H(A_{\overline{\mathcal{X}}}^{[k+1]} | S_{\mathcal{X}}, Q_{\overline{\mathcal{X}}}^{[k+1]}, W_{[1:k]}) \\
  &\overset{(\ref{eq:mar})}{=} H(A_{\overline{\mathcal{X}}}^{[k+1]} | S_{\mathcal{X}}, Q_{[1:N]}^{[k+1]}, W_{[1:k]}) \\
&\overset{(\ref{ans_det})(\ref{corr})}{=} H(A_{\overline{\mathcal{X}}}^{[k+1]}, W_{k+1} | S_{\mathcal{X}}, Q_{[1:N]}^{[k+1]}, W_{[1:k]})\\
&\overset{}{=} H(W_{k+1} | S_{\mathcal{X}}, Q_{[1:N]}^{[k+1]}, W_{[1:k]}) + H(A_{\overline{\mathcal{X}}}^{[k+1]} | S_{\mathcal{X}}, Q_{[1:N]}^{[k+1]}, W_{[1:k+1]}) \\
&\overset{(\ref{eq:ind1})}{=}  L + H(A_{\overline{\mathcal{X}}}^{[k+1]} | S_{\mathcal{X}}, Q_{[1:N]}^{[k+1]}, W_{[1:k+1]})
\end{align}
This completes the proof of Lemma \ref{lemma:recursive}.
$\hfill\square$

Now let us apply Lemma \ref{lemma:recursive} repeatedly for $k = 1, 2, \cdots$. When $N > X+T$, we have
\begin{align}
H(A_{\overline{\mathcal{X}}}^{[1]} | S_{\mathcal{X}}, Q_{[1:N]}^{[1]}, W_1) &\geq \frac{T}{N-X} \left( L + H(A_{\overline{\mathcal{X}}}^{[2]} | S_{\mathcal{X}}, Q_{[1:N]}^{[2]}, W_{[1:2]} \right) \\
&\geq \frac{T}{N-X}\left( L + \frac{T}{N-X}\left(L+H(A_{\overline{\mathcal{X}}}^{[3]} | S_{\mathcal{X}}, Q_{[1:N]}^{[3]}, W_{[1:3]}\right) \right) \\
&\geq \cdots\\
&\geq L\left( \frac{T}{N-X} + \left(\frac{T}{N-X}\right)^2 + \cdots + \left(\frac{T}{N-X}\right)^{K-1} \right) \label{eq:in1}
\end{align}
Similarly, when $N \leq X+T$, we have
\begin{align}
H(A_{\overline{\mathcal{X}}}^{[1]} | S_{\mathcal{X}}, Q_{[1:N]}^{[1]}, W_1) &\geq  L + H(A_{\overline{\mathcal{X}}}^{[2]} | S_{\mathcal{X}}, Q_{[1:N]}^{[2]}, W_{[1:2]}) \\
&\geq \cdots \\
&\geq L(K-1) \label{eq:in2}
\end{align}

\noindent Substituting (\ref{eq:in1}), (\ref{eq:in2}) into (\ref{eq:int}), we have
\begin{align}
& L \leq \sum_{n \in \overline{\mathcal{X}}}D_n - L\left( \frac{T}{N-X} + \left(\frac{T}{N-X}\right)^2 + \cdots + \left(\frac{T}{N-X}\right)^{K-1} \right), ~\mbox{if}~N > X+T.\\
& L \leq \sum_{n \in \overline{\mathcal{X}}}D_n - L(K-1),~\mbox{if}~N \leq X+T.
\end{align}
Averaging over all $\mathcal{X}$, we have
\begin{align}
& L \leq  \left(\frac{N-X}{N} \right)D - L\left( \frac{T}{N-X} + \left(\frac{T}{N-X}\right)^2 + \cdots + \left(\frac{T}{N-X}\right)^{K-1} \right), ~\mbox{if}~N > X+T.\\
& L \leq  \left(\frac{N-X}{N} \right)D - L(K-1),~\mbox{if}~N \leq X+T.
\end{align}
Finally since the rate  is defined as $R = L/D$, we arrive at the final bound.
\begin{align}
& R \leq  \frac{N-X}{N} \left(1+ \frac{T}{N-X} + \left(\frac{T}{N-X}\right)^2 + \cdots + \left(\frac{T}{N-X}\right)^{K-1} \right)^{-1}, ~\mbox{if}~N > X+T.\\
& R \leq  \frac{N-X}{N} \times \frac{1}{K} ,~\mbox{if}~N \leq X+T. \\
\intertext{Thus}
& C_{\mbox{\tiny XSTPIR}}(N, K, X, T)\leq \left(\frac{N-X}{N}\right)C_{\mbox{\tiny TPIR}}(N-X, K, T),
\end{align}
and the proof of Theorem \ref{theorem:converse} is complete.
$\hfill \square$

\section{Proof of Theorem \ref{theorem:mds}}\label{proof:mds}
Let each message consist of $L=N-X$ symbols in $\mathbb{F}_q$, $q\geq N$, and append $X$ instances of $0$ symbols, to create artificial messages of length $N$, 
\begin{align}
\bar{W}_k&=(W_{k1}, W_{k_2}, \cdots, W_{k(N-X)}, \underbrace{0, 0, \cdots, 0}_{X}), && \forall k\in[1:K].
\end{align}
Corresponding to each message $W_k$, let $Z_k=(Z_{k1}, Z_{k2}, \cdots, Z_{kX})\in\mathbb{F}_q^{X}$ be $X$ independent uniform noise symbols,  to be used for $X$-security.  Let $Z_k$ be encoded with an $(N, X)$ MDS code to produce $\bar{Z}_k\in\mathbb{F}_q^N$. For each $k\in[1:K]$ and $n\in[1:N]$, the $n^{th}$ server stores the $n^{th}$ symbol of $\bar{W}_k+\bar{Z}_k$. Thus, each server stores a total of $K$ symbols. The MDS property of $\bar{Z}_k$ ensures that the data storage is $X$-secure. Retrieval is trivial --- in order to retrieve the desired message $W_\theta$, the  user simply downloads everything from all servers. Since the queries do not depend on the desired message, the scheme is $N$-private, so it is also $T$-private. The rate achieved is $\frac{N-X}{NK}$ which matches the capacity for this setting.$\hfill\square$

\section{Proof of Theorem \ref{theorem:main}}\label{proof:main}
Let us present an  XSTPIR scheme for arbitrary $X$, $T$, $N$, $K$, that is asymptotically optimal (as $K\rightarrow\infty$). The  asymptotic capacity is zero for $N\leq X+T$, so we only need to consider $N>X+T$. 
Throughout this scheme we will set
\begin{align}
L&=N-X-T
\end{align}
and we will use the compact notation,
\begin{align}
\Delta&=\prod_{i=1}^{L}(i+\alpha).\label{eq:delta}
\end{align}
$\Delta_n$ will represent the value of $\Delta$ when $\alpha$ is replaced with $\alpha_n$.

Each message $W_k, k\in[1:K]$, consists of $L=N-X-T$ symbols, $W_k=(W_{k1}, W_{k2}, \cdots, W_{kL})$ from a finite field $\mathbb{F}_q$. The field $\mathbb{F}_q$ is assumed to have size $q\geq L+N$, and  characteristic greater than $L-1$. For the design of this scheme, we will need constants $\alpha_n, n\in[1:N]$ that are distinct elements of $\mathbb{G}$,
\begin{align}
\mathbb{G}&=\{\alpha\in\mathbb{F}_q: \alpha+i\neq 0, ~\forall i\in[1:L]\}. \label{eq:gg}
\end{align}
Such $\alpha_n, n\in[1:N]$ must exist because $q\geq L+N$. These constants will be globally known. In the following description of the scheme, we will explain explicitly how the values of these constants are chosen. For now, let us note that because the characteristic of the field is assumed to be greater than $L-1$, the values $\alpha+1, \alpha+2, \cdots, \alpha+L$ are distinct for any $\alpha\in\mathbb{F}_q$.

Let us split the messages into $L$ vectors, so that ${\bf W}_l=(W_{1l}, W_{2l}, \cdots, W_{Kl})$, $l\in[1:L]$, contains the $l^{th}$ symbol of every message. Let ${\bf Z}_{lx}, l\in[1:L], x\in[1:X]$, be independent uniformly random noise vectors from $\mathbb{F}_q^{1\times K}$, that are used to guarantee security. Similarly, let ${\bf Z}'_{lt}, l\in[1:L], t\in[1:T]$, be independent uniformly random noise vectors from $\mathbb{F}_q^{K\times 1}$, that are used to guarantee privacy. The independence between noise vectors, messages, and the user's desired message index $\theta$ is specified as follows.
\begin{align}
H\Big(\left({\bf W}_l\right)_{l\in[1:L]}, \left({\bf Z}_{lx}\right)_{l\in[1:L], x\in[1:X]}, \left({\bf Z}'_{lt}\right)_{l\in[1:L], t\in[1:T]},\theta\Big)&=H(\left({\bf W}_l\right)_{l\in[1:L]})+H(\theta)+KL(X+T)
\end{align}
in $q$-ary units. Let ${\bf Q}_\theta$ represent\footnote{Note that the XSTPIR scheme described in this section works even if ${\bf Q}_\theta$ is an arbitrary vector, i.e., if instead of retrieving one of the $K$ messages, the user wishes to compute an arbitrary linear function of the $K$ messages over $\mathbb{F}_q$. Thus, the scheme automatically settles the asymptotic capacity of the natural $X$-secure  and $T$-private generalization of the linear private computation problem introduced in \cite{Sun_Jafar_PC} (also known as linear private function retrieval \cite{Mirmohseni_Maddah}).} the $\theta^{th}$ column of the $K\times K$ identity matrix, so it contains a $1$ in the $\theta^{th}$ position and zeros everywhere else. Note that 
\begin{align}
({\bf W}_1{\bf Q}_\theta, {\bf W}_2{\bf Q}_\theta, \cdots, {\bf W}_L{\bf Q}_\theta)&=(W_{\theta1}, W_{\theta2}, \cdots, W_{\theta L})=W_\theta
\end{align}
is the message desired by the user. A succinct summary of the storage at each server, the queries, and a partitioning of signal and interference dimensions contained in the answers from each server, is provided below.

\begin{align*}
\begin{array}{cc}\hline
&\mbox{Server `$n$' (Replace $\alpha,\Delta$ with $\alpha_n,\Delta_n$)} \\\hline
\mbox{Storage} & \mathbf{W}_1+(1+\alpha)\mathbf{Z}_{11}+\cdots+(1+\alpha)^X\mathbf{Z}_{1X},\\
(S_n)&\mathbf{W}_2+(2+\alpha)\mathbf{Z}_{21}+\cdots+(2+\alpha)^X\mathbf{Z}_{2X},\\
&\vdots\\
&\mathbf{W}_{L}+(L+\alpha)\mathbf{Z}_{L1}+\cdots+(L+\alpha)^X\mathbf{Z}_{LX}\\\hline
\mbox{Query} &\frac{\Delta}{1+\alpha}\Big(\mathbf{Q_\theta}+(1+\alpha)\mathbf{Z}_{11}'+\cdots+(1+\alpha)^T\mathbf{Z}_{1T}'\Big),\\
(Q_n^{[\theta]})&\frac{\Delta}{2+\alpha}\Big(\mathbf{Q_\theta}+(2+\alpha)\mathbf{Z}_{21}'+\cdots+(2+\alpha)^T\mathbf{Z}_{2T}'\Big),\\
&\vdots\\
&\frac{\Delta}{L+\alpha}\Big(\mathbf{Q_\theta}+(L+\alpha)\mathbf{Z}_{L1}'+\cdots+(L+\alpha)^T\mathbf{Z}_{LT}'\Big)\\\hline
\multicolumn{2}{c}{\mbox{Desired symbols appear along vectors}}\\
\multicolumn{2}{c}{\vect{\Delta}\circ\Big(\vect{(1+\alpha)^{-1}}, \vect{(2+\alpha)^{-1}}, \cdots, \vect{(L+\alpha)^{-1}}\Big)}\\
\multicolumn{2}{c}{\mbox{Interference  appears along vectors}}\\
\multicolumn{2}{c}{\vect{\Delta}\circ\Big(\vect{1}, \vect{(1+\alpha)},  \cdots, \vect{(1+\alpha)^{ X+T-1}}, \vect{(2+\alpha)}, \cdots, \vect{(2+\alpha)^{ X+T-1}}, \cdots,} \\
\multicolumn{2}{c} {\hspace{1cm}\cdots, \vect{(L+\alpha)}, \cdots, \vect{(L+\alpha)^{ X+T-1}}\Big)}\\\hline
\end{array}
\end{align*}

Initially, the user knows only his desired message index $\theta$ and the noise terms ${\bf Z}'_{lt}, l\in[1:L], t\in[1:T]$, all of which are privately generated by the user. Each server $n\in [1:N]$ knows only its stored information $S_n$. The storage $S_n$ at Server $n$ may be viewed as a $1\times LK$ row vector formed by concatenating the $L$ row vectors, ${\bf W}_l+\sum_{x=1}^X(l+\alpha_n)^x{\bf Z}_{lx}$, $l\in[1:L]$. Similarly, the query $Q_n^{[\theta]}$ may be viewed as an $LK\times 1$ column vector formed by concatenating the $L$ column vectors, $\frac{\Delta_n}{l+\alpha_n}\left({\bf Q}_\theta + \sum_{t=1}^T(l+\alpha_n)^t{\bf Z}_{lt}'\right)$, $l\in[1:L]$.

Upon receiving the query $Q^{[\theta]}_n$ from the user, Server $n$ responds with the answer $A_n^{[\theta]}$ that is exactly one symbol in $\mathbb{F}_q$, found by multiplying  $S_n$ with $Q^{[\theta]}_n$.
\begin{align}
A_n^{[\theta]}&=S_nQ^{[\theta]}_n .
\end{align}
This produces a single equation in a total of $L(X+1)(T+1)$ terms. Out of these, $L$ terms are desired message symbols ${\bf W}_l{\bf Q}_\theta$, $l\in[1:L]$, and the remaining $L(X+1)(T+1)-L$ terms are undesired, or interference terms. The interference terms include $LT$ terms of the type ${\bf W}_l {\bf Z}_{lt}'$, $LX$ terms of the type ${\bf Z}_{lx}{\bf Q}_\theta$, and $LXT$ terms of the type ${\bf Z}_{lx}{\bf Z}_{lt}'$. The user obtains one such equation from each server, for a total of $N$ equations, from which he must be able to retrieve his $L$ desired symbols. The key to this is the alignment of $L(X+1)(T+1)-L$ interference terms into $N-L$ dimensions, leaving $L$ dimensions free from interference from which the $L$ desired symbols can be decoded.

First let us identify the desired signal dimensions, i.e., the vectors along which desired symbols are seen by the user.  Each answer $A_n^{[\theta]}$ contains the desired symbols $\frac{\Delta_n}{l+\alpha_n}{\bf W}_l{\bf Q}_\theta = \frac{\Delta_n}{l+\alpha_n}W_{\theta l}$, $l\in[1:L]$. These $L$ desired symbols appear along the following $L$ vectors.
\begin{align}
\left[
\begin{matrix}
\frac{\Delta_1}{1+\alpha_1}\\
\frac{\Delta_2}{1+\alpha_2}\\
\vdots\\
\frac{\Delta_N}{1+\alpha_N}\\
\end{matrix}
\right], \left[
\begin{matrix}
\frac{\Delta_1}{2+\alpha_1}\\
\frac{\Delta_2}{2+\alpha_2}\\
\vdots\\
\frac{\Delta_N}{2+\alpha_N}\\
\end{matrix}
\right], \cdots, \left[
\begin{matrix}
\frac{\Delta_1}{L+\alpha_1}\\
\frac{\Delta_2}{L+\alpha_2}\\
\vdots\\
\frac{\Delta_N}{L+\alpha_N}\\
\end{matrix}
\right]\triangleq 
\vect{\Delta}\circ\Big(\vect{(1+\alpha)^{-1}}, \vect{(2+\alpha)^{-1}}, \cdots, \vect{(L+\alpha)^{-1}}\Big).
\end{align}
Recall that $\circ$ represents the Hadamard product. 
Similarly, the vectors along which interference symbols appear are identified as follows.
\begin{align}
{\vect{\Delta}\circ\Big(\vect{1}, \vect{(1+\alpha)},  \cdots, \vect{(1+\alpha)^{ X+T-1}}, \vect{(2+\alpha)}, \cdots, \vect{(2+\alpha)^{ X+T-1}}, \cdots,}\nonumber\\
 {\hspace{1cm}\cdots, \vect{(L+\alpha)}, \cdots, \vect{(L+\alpha)^{ X+T-1}}\Big)} .
\end{align}
Thus, the vector of answers from all $N$ servers can be expressed as
\begin{align}
\vect{A^{[\theta]}}&=\sum_{l=1}^{L}W_{\theta l}\vect{\Delta}\circ\vect{(l+\alpha)^{-1}}+\sum_{l=1}^L\sum_{i=0}^{X+T-1}\vect{\Delta}\circ\vect{(l+\alpha)^{ i}}I_{li}
\end{align}
for some interference terms $I_{li}$ that are sums of various ${\bf W}_l {\bf Z}_{lt}'$,  ${\bf Z}_{lx}{\bf Q}_\theta$, and ${\bf Z}_{lx}{\bf Z}_{lt}'$ terms. The exact form of $I_{li}$ terms  is not important for our analysis. Using binomial expansion to write each $\vect{(l+\alpha)^{ i}}$ vector as $\sum_{j=0}^{i}\binom{i}{j}l^j\vect{\alpha^{ i-j}}$, and grouping terms by the vectors $\vect{\alpha^{ i}}$, we can write,
\begin{align}
\vect{A^{[\theta]}}&=\sum_{l=1}^{L}W_{\theta l}\vect{\Delta}\circ\vect{(l+\alpha)^{-1}}+
\sum_{i=0}^{X+T-1}\vect{\Delta}\circ\vect{\alpha^{ i}} I_{i}' .
\end{align}
Thus,  all interference is aligned within the subspace spanned by vectors $\vect{\Delta}$, $\vect{\Delta}\circ\vect{\alpha}$, $\dots$, ${\vect{\Delta}\circ\vect{\alpha^{ X+T-1}}}$. As explained in Section \ref{sec:cross}, this is because of \emph{cross subspace alignment}.

In matrix notation, we have,
\begin{align}
\vect{A^{[\theta]}}=\left[
\begin{matrix}
A_1^{[\theta]}\\
A_2^{[\theta]}\\
\vdots\\
A_N^{[\theta]}\\
\end{matrix}
\right]
&=
M_N
\left[\begin{matrix} 
W_{\theta1}\\
\vdots\\
W_{\theta L}\\
I_{0}'\\
\vdots\\
I_{(X+T-1)}'
\end{matrix}
\right]
\end{align}
where the $N\times N$ square matrix (note that $L+X+T=N$)
\begin{align}
M_N&=\left[\begin{matrix}
\frac{\Delta_1}{1+\alpha_1}&\cdots&\frac{\Delta_1}{L+\alpha_1}&\Delta_1&\Delta_1\alpha_1&\cdots&\Delta_1\alpha_1^{X+T-1}\\
\frac{\Delta_2}{1+\alpha_2}&\cdots&\frac{\Delta_2}{L+\alpha_2}&\Delta_2&\Delta_2\alpha_2&\cdots&\Delta_2\alpha_2^{X+T-1}\\
\vdots\\
\frac{\Delta_N}{1+\alpha_N}&\cdots&\frac{\Delta_N}{L+\alpha_N}&\Delta_N&\Delta_N\alpha_N&\cdots&\Delta_N\alpha_N^{X+T-1}
\end{matrix}
\right]\\
&=\left[\begin{matrix}
\vect{\Delta}\circ\vect{(1+\alpha)^{-1}}&\cdots&\vect{\Delta}\circ\vect{(L+\alpha)^{-1}}&\vect{\Delta}&\vect{\Delta}\circ\vect{\alpha}&\cdots&\vect{\Delta}\circ\vect{\alpha^{ X+T-1}}
\end{matrix}
\right]
\end{align}
is called the decoding matrix. Evidently, if the decoding matrix is invertible, then  the user can recover his $L$ desired message symbols. We show that if $\alpha_n, n\in[1:N]$ are distinct elements of $\mathbb{G}$, then $M_N$ is invertible. This result is stated in the following lemma. Note that in our design, we have chosen $\alpha_n$ as distinct elements, so Lemma \ref{disinv} guarantees that the scheme satisfies the correctness constraint. Fixing distinct values of $\alpha_1, \cdots, \alpha_N$ completes the design of the scheme. 

\begin{lemma}\label{disinv}
The decoding matrix $M_N$ is invertible if all $\alpha_n, n\in[1:N]$ are distinct.
\end{lemma}
\begin{proof}
To set up the proof by contradiction, suppose on the contrary that $M_N$ is singular. Then there must exist $c_n \in \mathbb{F}_q, n\in[1:N]$, at least one of which is non-zero, such that 
\begin{equation}
c_1\vect{\Delta}\circ\vect{(1+\alpha)^{-1}}+\cdots+c_L\vect{\Delta}\circ\vect{(L+\alpha)^{-1}}+c_{L+1}\vect{\Delta}+c_{L+2}\vect{\Delta}\circ\vect{\alpha}+\cdots+c_N\vect{\Delta}\circ\vect{\alpha^{ X+T-1}}=\vect{0} \label{eq:de}
\end{equation}
where $\vect{0}$ is the vector whose elements are all 0.
Now consider $n$-th row of (\ref{eq:de}).
\begin{equation}
c_1\frac{\Delta_n}{1+\alpha_n}+\cdots+c_L\frac{\Delta_n}{L+\alpha_n}+c_{L+1}\Delta_n+c_{L+2}\Delta_n\alpha_n+\cdots+c_N\Delta_n\alpha_n^{X+T-1}=0.
\end{equation}
From (\ref{eq:delta}) and (\ref{eq:gg}), we know that $\Delta_n\neq0$. Then $\alpha_n$ must be the root of the following polynomial
\begin{align}
g(\alpha) &= \sum_{i=1}^L c_i\left(\frac{\Delta}{i+\alpha}\right) + \sum_{i=L+1}^N c_i \Delta \alpha^{i-(L+1)}
\end{align}
Note that $\Delta$ (as a function of $\alpha$) has order $L$ and $i+\alpha$ is a factor of $\Delta$ (refer to (\ref{eq:delta})), so $g(\alpha)$ has order \emph{at most} $N-1$. If $g(\alpha)$ is a non-zero polynomial, then it can have at most $N-1$ roots over $\mathbb{F}_q$. Now $\alpha_n, n\in[1:N]$ are $N$ distinct roots of $g(\alpha)$, thus $g(\alpha)$ must be the zero polynomial, i.e., the coefficients of all monomials in $g(\alpha)$ must be zero. The coefficient of $\alpha^{N-1}$ is $c_N$ so we must have $c_N=0$. Then, the remaining coefficient of $\alpha^{N-2}$ is $c_{N-1}$, so we must have $c_{N-1}=0$. Similarly, we find $c_{L+1}=c_{L+2}=\cdots = c_N=0$, leaving us with 
\begin{align}
g(\alpha) &= \sum_{i=1}^L c_i\left(\frac{\Delta}{i+\alpha}\right) .
\end{align}
Now, if this $g(\alpha)$ is the zero polynomial, then it must be zero for every $\alpha \in \mathbb{F}_q$. Choosing $\alpha$ such\footnote{Note that $\frac{\Delta}{i+\alpha}$ is simply a compact notation for $\prod_{l\in[1:L], l\neq i}(l+\alpha)$, i.e., it only means that the $(i+\alpha)$ factor is eliminated from $\Delta$, so there is no `division by $0$' when we set $i+\alpha=0$ in $\frac{\Delta}{i+\alpha}$.} that $(i+\alpha)=0$, gives us $c_i=0$ for every $i\in[1:L]$. Thus, we have $c_1=c_2=\cdots=c_N=0$. This is a contradiction since we assumed that at least one of $c_n, n\in[1:N]$ is non-zero. Thus, the proof is complete. $\hfill\square$
\end{proof}

Now consider the security guarantee. For any $X$ colluding servers, $i_1, i_2, \cdots, i_X$, the $X$ observations, $U_{kl1}, \cdots, U_{klX}$, of each message symbol $W_{kl}, k\in[1:K], l\in[1:L]$, are protected by  noise terms as follows.
\begin{align}
\left[
\begin{matrix}
U_{kl1}\\
\vdots\\
U_{klX}
\end{matrix}
\right]&=\left[
\begin{matrix}
W_{kl}\\
\vdots\\
W_{kl}
\end{matrix}
\right]+\underbrace{\left[
\begin{matrix}
l+\alpha_{i_1} & (l+\alpha_{i_1})^2 &\cdots& (l+\alpha_{i_1})^X \\
l+\alpha_{i_2} & (l+\alpha_{i_2})^2 &\cdots& (l+\alpha_{i_2})^X\\
\vdots&\vdots&&\vdots\\
l+\alpha_{i_X} & (l+\alpha_{i_X})^2 &\cdots& (l+\alpha_{i_X})^X
\end{matrix}
\right]}_{P}\underbrace{\left[
\begin{matrix}
{\bf Z}_{l1}(k)\\
\vdots\\
{\bf Z}_{lX}(k)
\end{matrix}
\right]}_{Z}\\
&=W_{kl}\underbrace{\left[
\begin{matrix}
1\\
\vdots\\
1
\end{matrix}
\right]}_{\bf 1}+
\left[\begin{matrix}
l+\alpha_{i_1}&0&\cdots&0\\
0&l+\alpha_{i_2}&\cdots&0\\
0&0&\ddots&0\\
0&0&\cdots&l+\alpha_{i_X}
\end{matrix}
\right]\left[
\begin{matrix}
1& l+\alpha_{i_1} &\cdots& (l+\alpha_{i_1})^{X-1} \\
1& l+\alpha_{i_2} &\cdots& (l+\alpha_{i_2})^{X-1}\\
\vdots&\vdots&&\vdots\\
1 & l+\alpha_{i_X} &\cdots& (l+\alpha_{i_X})^{X-1}
\end{matrix}
\right]\left[
\begin{matrix}
{\bf Z}_{l1}(k)\\
\vdots\\
{\bf Z}_{lX}(k)
\end{matrix}
\right] .
\end{align}
where ${\bf Z}_{lx}(k)$ is the $k^{th}$ element of the vector ${\bf Z}_{lx}$.
Note that $P$ is a product of a diagonal matrix which is invertible because  $(l+\alpha_{i_j})$ are non-zero, and a Vandermonde matrix which is invertible because $(l+\alpha_{i_j})$ are distinct. Therefore, $P$ is invertible, and the observations are independent of the message symbols as shown below.
\begin{align}
I(W_{kl}; \left(U_{klx}\right)_{x\in[1:X]})&=I(W_{kl}; W_{kl}{\bf 1}+PZ)=I(W_{kl}; W_{kl}P^{-1}{\bf 1}+Z)=I(W_{kl}; Z)=0.
\end{align}
Furthermore, since the noise terms protecting each message symbol $W_{kl}, k\in[1:K], l\in[1:L]$, i.e., ${\bf Z}_{lx}(k), x\in[1:X]$ are independent across $(k,l,x)$, security is preserved for all data.

The noise terms protecting each \emph{query} also have the same structure and independence properties by design. Therefore, it follows from the same reasoning that user's privacy is protected from any $T$ colluding servers.

Finally, note that the user is able to retrieve $L=N-X-T$ desired $q$-ary symbols by downloading $N$ $q$-ary symbols, one from each server. The rate achieved is $L/N=1-(X+T)/N$, which is the asymptotic capacity for this general setting. This completes the proof of Theorem \ref{theorem:main}. $\hfill\square$

\subsection{Example: $(X=1)$ Secure, $(T=1)$ Private Scheme with $N=5$ Servers}
Each message consists of $L=3$ symbols from a finite field $\mathbb{F}_q$, $q\geq N+L=8$, and characteristic greater than 2.  For this setting, $\Delta=(1+\alpha)(2+\alpha)(3+\alpha)$.
\begin{align*}
\begin{array}{cc}\hline
&\mbox{Server `$n$' (Replace $\alpha,\Delta$ with $\alpha_n,\Delta_n$)} \\\hline
\mbox{Storage} & \mathbf{W}_1+(1+\alpha)\mathbf{Z}_{1},\\
(S_n) & \mathbf{W}_2+(2+\alpha)\mathbf{Z}_{2},\\
 & \mathbf{W}_3+(3+\alpha)\mathbf{Z}_{3}\\\hline
\mbox{Query} &\frac{\Delta}{1+\alpha}\Big(\mathbf{Q_\theta}+(1+\alpha)\mathbf{Z}_{1}'\Big),\\
(Q_n^{[\theta]}) &\frac{\Delta}{2+\alpha}\Big(\mathbf{Q_\theta}+(2+\alpha)\mathbf{Z}_{2}'\Big),\\
&\frac{\Delta}{3+\alpha}\Big(\mathbf{Q_\theta}+(3+\alpha)\mathbf{Z}_{3}'\Big)\\\hline
\multicolumn{2}{c}{\mbox{Desired symbols appear along vectors}}\\
\multicolumn{2}{c}{\vect{\Delta}\circ\Big(\vect{(1+\alpha)^{-1}}, \vect{(2+\alpha)^{-1}}, \vect{(3+\alpha)^{-1}}\Big)}\\
\multicolumn{2}{c}{\mbox{Interference symbols appear along vectors}}\\
\multicolumn{2}{c}{\vect{\Delta}\circ\Big(\vect{1}, \vect{1+\alpha}, \vect{2+\alpha}, \vect{3+\alpha}\Big)}\\\hline
\end{array}
\end{align*}
The answers from all $N=5$ servers may be written explicitly as,
\begin{align}
\vect{A^{[\theta]}}&=\vect{\Delta}\circ\vect{(1+\alpha)^{-1}}{\bf W}_1{\bf Q}_\theta +\vect{\Delta}\circ\vect{(2+\alpha)^{-1}}{\bf W}_2{\bf Q}_\theta +
\vect{\Delta}\circ\vect{(3+\alpha)^{-1}}{\bf W}_3{\bf Q}_\theta\nonumber\\
&\hspace{1cm} +
\vect{\Delta} \underbrace{\left({\bf W}_1{\bf Z}_1'+{\bf W}_2{\bf Z}_2'+{\bf W}_3{\bf Z}_3'+{\bf Z}_{1}{\bf Q}_\theta+{\bf Z}_{2}{\bf Q}_\theta+{\bf Z}_{3}{\bf Q}_\theta\right)}_{I_{10}+I_{20}+I_{30}}\nonumber\\
&\hspace{1cm} +\vect{\Delta}\circ\vect{(1+\alpha)} \underbrace{\left({\bf Z}_{1}{\bf Z}_1'\right)}_{I_{11}}+\vect{\Delta}\circ\vect{(2+\alpha)} \underbrace{\left({\bf Z}_{2}{\bf Z}_2'\right)}_{I_{21}}+\vect{\Delta}\circ\vect{(3+\alpha)} \underbrace{\left({\bf Z}_{3}{\bf Z}_3'\right)}_{I_{31}}\\
&=\vect{\Delta}\circ\vect{(1+\alpha)^{-1}}W_{\theta 1} +\vect{\Delta}\circ\vect{(2+\alpha)^{-1}}W_{\theta 2} +
\vect{\Delta}\circ\vect{(3+\alpha)^{-1}}W_{\theta 3}\nonumber\\
&\hspace{1cm}  +\vect{\Delta}(I_{10}+I_{20}+I_{30}+I_{11}+2I_{21}+3I_{31})+\vect{\Delta}\circ\vect{\alpha}(I_{11}+I_{21}+I_{31}) .
\end{align}
Privacy and security are guaranteed since $1+\alpha_n\neq0, \forall n\in[1:5]$, the messages and queries are hidden behind the noise. 

Interference terms align into the space spanned by the two vectors, $\vect{\Delta}, \vect{\Delta}\circ\vect{\alpha}$, while the $3$ symbols of the desired message appear along $\vect{\Delta}\circ\vect{(1+\alpha)^{-1}}, \vect{\Delta}\circ\vect{(2+\alpha)^{-1}}, \vect{\Delta}\circ\vect{(3+\alpha)^{-1}}$.  Independence of the $3$ desired signal dimensions from the two  interference dimensions is trivially verified, because the highest exponent of $\alpha$ along desired signal dimensions is $2$, but each interference dimension has an $\alpha^3$ term (contributed by $\Delta$). Independence of the $3$ desired signal dimensions among themselves is also easily verified, because for
\begin{align}
c_1(2+\alpha)(3+\alpha)+c_2(1+\alpha)(3+\alpha)+c_3(1+\alpha)(2+\alpha)
\end{align}
to be the zero polynomial it must be zero everywhere, but in that case, setting $\alpha+i=0$ for $i=1,2,3$, leads us to  $c_1=c_2=c_3=0$, thus proving their independence. The rate achieved is $3/5$, which matches the asymptotic capacity for this setting.
\subsection{Example: $(X=2)$ Secure, $(T=1)$ Private Scheme with $N=4$ Servers}
Each message consists of $L=1$ symbol from a finite field $\mathbb{F}_q$, $q\geq N+L=5$.  $\Delta=(1+\alpha)$.
\begin{align*}
\begin{array}{cc}\hline
&\mbox{Server `$n$' (Replace $\alpha,\Delta$ with $\alpha_n,\Delta_n$)} \\\hline
\mbox{Storage $(S_n)$} & \mathbf{W}_1+(1+\alpha)\mathbf{Z}_{11}+(1+\alpha)^2\mathbf{Z}_{12}\\\hline
\mbox{Query $(Q_n^{[\theta]})$} &\mathbf{Q_\theta}+(1+\alpha)\mathbf{Z}_1'\\\hline
\multicolumn{2}{c}{\mbox{Desired symbols appear along vector}}\\
\multicolumn{2}{c}{\vect{1}}\\
\multicolumn{2}{c}{\mbox{Interference symbols appear along vectors}}\\
\multicolumn{2}{c}{\vect{\Delta}\circ\Big(\vect{1}, \vect{1+\alpha}, \vect{(1+\alpha)^{2}}\Big)}\\\hline
\end{array}
\end{align*}
The answers from all $N=4$ servers may be written explicitly as,
\begin{align}
\vect{A^{[\theta]}}&=\vect{1}{\bf W}_1{\bf Q}_\theta +\vect{(1+\alpha)} \left({\bf W}_1{\bf Z}_1'+{\bf Z}_{11}{\bf Q}_\theta\right)+\vect{(1+\alpha)^{2}} \left({\bf Z}_{11}{\bf Z}_1'+{\bf Z}_{12}{\bf Q}_\theta\right)+\vect{(1+\alpha)^{3}}{\bf Z}_{12}{\bf Z}_1'\\
&=\vect{1}W_{\theta 1} +\vect{\Delta}\circ
\left( \vect{1}
\underbrace{\left({\bf W}_1{\bf Z}_1'+{\bf Z}_{11}{\bf Q}_\theta\right)}_{I_{10}}+\vect{(1+\alpha)} \underbrace{\left({\bf Z}_{11}{\bf Z}_1'+{\bf Z}_{12}{\bf Q}_\theta\right)}_{I_{11}}+\vect{(1+\alpha)^{2}}\underbrace{{\bf Z}_{12}{\bf Z}_1'}_{I_{12}}\right)\\
&=\vect{1}W_{\theta 1} +\vect{\Delta}\underbrace{\left(I_{10}+I_{11}+I_{12}\right)}_{I_{0}'}+\vect{\Delta}\circ\vect{\alpha}\underbrace{(I_{11}+2I_{12})}_{I_{1}'}+\vect{\Delta}\circ\vect{\alpha^{2}}\underbrace{(I_{12})}_{I_2'} .
\end{align}
Interference aligns in the space spanned by the three vectors, $\vect{\Delta}, \vect{\Delta}\circ\vect{\alpha}, \vect{\Delta}\circ\vect{\alpha^{2}}$, while the desired symbol appears along the vector of all ones. The independence of these directions is easily established. Privacy is guaranteed because  $ 1+\alpha_n\neq 0$, $\forall n\in[1:4]$, so the queries are hidden behind random noise. Security is guaranteed because for any $X=2$ colluding servers, $i$ and $j$, the independent noise protecting each message $W_{k}$, $k\in[1:K]$, 
\begin{align}
\underbrace{\left[
\begin{matrix}
1+\alpha_i & (1+\alpha_i)^2\\
1+\alpha_j & (1+\alpha_j)^2
\end{matrix}
\right]}_{P_{ij}}\left[
\begin{matrix}
{\bf Z}_{11}(k)\\
{\bf Z}_{12}(k)
\end{matrix}
\right]&=
\left[\begin{matrix}
1+\alpha_i&0\\
0&1+\alpha_j
\end{matrix}
\right]\left[
\begin{matrix}
1 & (1+\alpha_i)\\
1 & (1+\alpha_j)
\end{matrix}
\right]\left[
\begin{matrix}
{\bf Z}_{11}(k)\\
{\bf Z}_{12}(k)
\end{matrix}
\right]
\end{align}
spans $X=2$ dimensions, because $P_{ij}$  is invertible for distinct and non-zero values of $(1+\alpha_i), (1+\alpha_j)$.
The rate achieved is $1/4$ which matches the asymptotic capacity for this setting.

\subsection{Example: $(X=1)$ Secure, $(T=2)$ Private Scheme with $N=5$ Servers}
Each message consists of $L=N-X-T=2$ symbols from $\mathbb{F}_q$, $q\geq 7$.
\begin{align}
\Delta&=(1+\alpha)(2+\alpha).
\end{align}
\begin{align*}
\begin{array}{cc}\hline
&\mbox{Server `$n$' (Replace $\alpha,\Delta$ with $\alpha_n,\Delta_n$)} \\\hline
\mbox{Storage} & \mathbf{W}_1+(1+\alpha)\mathbf{Z}_{1},\\
(S_n) &\mathbf{W}_2+(2+\alpha)\mathbf{Z}_{2}\\\hline
\mbox{Query} &\frac{\Delta}{1+\alpha}\Big(\mathbf{Q_\theta}+(1+\alpha)\mathbf{Z}_{11}'+(1+\alpha)^2\mathbf{Z}_{12}'\Big),\\
(Q_n^{[\theta]}) &\frac{\Delta}{2+\alpha}\Big(\mathbf{Q_\theta}+(2+\alpha)\mathbf{Z}_{21}'+(2+\alpha)^2\mathbf{Z}_{22}'\Big)\\\hline
\multicolumn{2}{c}{\mbox{Desired symbols appear along vectors}}\\
\multicolumn{2}{c}{\vect{\Delta}\circ\Big(\vect{(1+\alpha)^{-1}}, \vect{(2+\alpha)^{-1}}\Big)}\\
\multicolumn{2}{c}{\mbox{Interference symbols appear along vectors}}\\
\multicolumn{2}{c}{\vect{\Delta}\circ\Big(\vect{1}, \vect{1+\alpha}, \vect{(1+\alpha)^{2}}, \vect{2+\alpha}, \vect{(2+\alpha)^{2}}\Big)}\\\hline
\end{array}
\end{align*}

The answers from all $N=5$ servers may be written explicitly as,
\begin{align}
\vect{A^{[\theta]}}&=\vect{\Delta}\circ\vect{(1+\alpha)^{-1}}{\bf W}_1{\bf Q}_\theta +\vect{\Delta}\circ\vect{(2+\alpha)^{-1}}{\bf W}_2{\bf Q}_\theta \\
&\hspace{1cm}+\vect{\Delta}
\underbrace{\left({\bf W}_1{\bf Z}_{11}'+{\bf W}_2{\bf Z}_{21}'+{\bf Z}_{1}{\bf Q}_\theta+{\bf Z}_{2}{\bf Q}_\theta\right)}_{I_{10}+I_{20}}+\vect{\Delta}\circ\vect{(2+\alpha)^{2}}\underbrace{{\bf Z}_{2}{\bf Z}_{22}'}_{I_{22}}\nonumber\\
&\hspace{1cm}+\vect{\Delta}\circ\vect{(1+\alpha)} \underbrace{\left({\bf Z}_{1}{\bf Z}_{11}'+{\bf W}_{1}{\bf Z}_{12}'\right)}_{I_{11}}+\vect{\Delta}\circ\vect{(2+\alpha)} \underbrace{\left({\bf Z}_{2}{\bf Z}_{21}'+{\bf W}_{2}{\bf Z}_{22}'\right)}_{I_{21}}+\vect{\Delta}\circ\vect{(1+\alpha)^{2}}\underbrace{{\bf Z}_{1}{\bf Z}_{12}'}_{I_{12}}\nonumber\\
&=\vect{\Delta}\circ\vect{(1+\alpha)^{-1}}W_{\theta 1} +\vect{\Delta}\circ\vect{(2+\alpha)^{-1}}W_{\theta 2} +\vect{\Delta}(I_{10}+I_{20}+I_{11}+2I_{21}+I_{12}+4I_{22})\nonumber\\
&\hspace{1cm}+\vect{\Delta}\circ\vect{\alpha}(I_{11}+I_{21}+2I_{12}+4I_{22})+\vect{\Delta}\circ\vect{\alpha^{2}}(I_{12}+I_{22}) .
\end{align}
Thus, interference aligns into the  space spanned by the $3$ vectors: $\vect{\Delta}, \vect{\Delta}\circ\vect{\alpha}, \vect{\Delta}\circ\vect{\alpha^{2}}$, while the $2$ desired symbols appear along $\vect{\Delta}\circ\vect{(1+\alpha)^{-1}}, \vect{\Delta}\circ\vect{(2+\alpha)^{-1}}$. Note that the highest exponent of $\alpha$ along a desired signal dimension is $1$, but every interference dimension contains $\alpha^2$ (contributed by $\Delta$), so the desired signals are independent of the interference. The independence of desired signals among themselves is also easily verified because if $c_1\frac{\Delta}{1+\alpha}+c_2\frac{\Delta}{2+\alpha}=c_1(2+\alpha)+c_2(1+\alpha)$ is the zero polynomial, then by substituting $i+\alpha=0$ for $i=1,2$ we find that we must have $c_1=c_2=0$. Privacy and security are guaranteed by the MDS coded independent noise terms mixed with the message and query symbols. The rate achieved is $2/5$, which matches the asymptotic capacity for this setting.

\subsection{Example: $(X=2)$ Secure, $(T=2)$ Private Scheme with $N=7$ Servers}
Each message consists of $L=3$ symbols from a finite field $\mathbb{F}_q$, of size $q\geq 10$ and characteristic greater than $2$. $$\Delta=(1+\alpha)(2+\alpha)(3+\alpha).$$

\begin{align*}
\begin{array}{cc}\hline
&\mbox{Server `$n$' (Replace $\alpha,\Delta$ with $\alpha_n, \Delta_n)$} \\\hline
\mbox{Storage} & \mathbf{W}_1+(1+\alpha)\mathbf{Z}_{11}+(1+\alpha)^2\mathbf{Z}_{12},\\
(S_n)&\mathbf{W}_2+(2+\alpha)\mathbf{Z}_{21}+(2+\alpha)^2\mathbf{Z}_{22},\\
&\mathbf{W}_3+(3+\alpha)\mathbf{Z}_{31}+(3+\alpha)^2\mathbf{Z}_{32}\\\hline
\mbox{Query} &\frac{\Delta}{1+\alpha}\Big(\mathbf{Q_\theta}+(1+\alpha)\mathbf{Z}_{11}'+(1+\alpha)^2\mathbf{Z}_{12}'\Big),\\
(Q_n^{[\theta]})&\frac{\Delta}{2+\alpha}\Big(\mathbf{Q_\theta}+(2+\alpha)\mathbf{Z}_{21}'+(2+\alpha)^2\mathbf{Z}_{22}'\Big),\\
&\frac{\Delta}{3+\alpha}\Big(\mathbf{Q_\theta}+(3+\alpha)\mathbf{Z}_{31}'+(3+\alpha)^2\mathbf{Z}_{32}'\Big)\\\hline
\multicolumn{2}{c}{\mbox{Desired symbols appear along vectors}}\\
\multicolumn{2}{c}{\vect{\Delta}\circ\Big(\vect{(1+\alpha)^{-1}}, \vect{(2+\alpha)^{-1}}, \vect{(3+\alpha)^{-1}}\Big)}\\
\multicolumn{2}{c}{\mbox{Interference symbols appear along vectors}}\\
\multicolumn{2}{c}{\vect{\Delta}\circ\Big(\vect{1}, \vect{1+\alpha}, \vect{(1+\alpha)^{2}}, \vect{(1+\alpha)^{3}}, \vect{2+\alpha}, \vect{(2+\alpha)^{2}}, \vect{(2+\alpha)^{3}}, }\\
\multicolumn{2}{r} {\hspace{1cm} \vect{3+\alpha}, \vect{(3+\alpha)^{2}}, \vect{(3+\alpha)^{3}}\Big)}\\\hline
\end{array}
\end{align*}
Interference aligns into the  space spanned by the $4$ vectors: $\vect{\Delta}, \vect{\Delta}\circ\vect{\alpha},  \vect{\Delta}\circ\vect{\alpha^{2}}, \vect{\Delta}\circ\vect{\alpha^{3}}$.
Independence of desired signals from interference is trivially verified -- highest exponent of $\alpha$ along any desired signal dimension is $2$, but each interference dimension has an $\alpha^3$ term (contributed by $\Delta$). The desired signal dimensions are easily verified to be linearly independent among themselves because in order for 
\begin{align}
c_1(2+\alpha)(3+\alpha)+c_2(1+\alpha)(3+\alpha)+c_3(1+\alpha)(2+\alpha)
\end{align}
to be the zero polynomial it must be zero everywhere, but in that case, setting $\alpha+i=0$ for $i=1,2,3$ leads us to  $c_1=c_2=c_3=0$. Privacy and security are guaranteed by the MDS coded independent noise terms mixed with the message and query symbols. 
The rate achieved is $3/7$, which matches the asymptotic capacity for this setting.

\section{Proof of Theorem \ref{theorem:N3}}\label{proof:N3}
In Section \ref{proof:main} we presented an XSTPIR scheme for arbitrary $X, T, N, K$ that achieves capacity as $K\rightarrow\infty$. Since the scheme also works for any $K$, a natural starting point for finite $K$ settings is to apply the same scheme. A key insight here is that the rate achieved by the scheme improves as $K$ decreases. Let us elaborate. Note that the query $Q_n^{[\theta]}$ that is sent to each server is uniformly distributed in $\mathbb{F}_q^{LK}$. Therefore, with probability $\frac{1}{q^{LK}}$, the query vector is the all zero vector. Whenever this happens, no download is needed from the server. Thus, the average download is reduced by the factor $(1-1/q^{LK})$ and the rate achieved is expressed as follows.
\begin{lemma}\label{lemma:finite}
The asymptotically capacity achieving XSTPIR scheme of Section \ref{proof:main} achieves the rate
\begin{align}
R&=\left(1-\frac{1}{q^{KL}}\right)^{-1}\left(1-\left(\frac{X+T}{N}\right)\right)
\end{align}
for arbitrary $X, L, K, N$ values, where $N>X+T$.
\end{lemma}
Note that $N\leq X+T$ is excluded as the degenerate setting where we already know the capacity for all parameters, according to Theorem \ref{theorem:mds}. Remarkably, the rate in Lemma \ref{lemma:finite}  depends on the message size $L$ and the field size $q$ used by the scheme. As presented, the scheme uses $q\geq L+N$ and $L=N-X-T$. So the achieved rate for finite $K$ becomes
\begin{align}
R&=\left(1-\frac{1}{(2N-X-T)^{K(N-X-T)}}\right)^{-1}\left(1-\left(\frac{X+T}{N}\right)\right).
\end{align}

Consider the simplest non-trivial setting of interest, i.e., the setting for Theorem \ref{theorem:N3}, where $T=X=1$, $N=3$ and $K$ is arbitrary. The scheme  of Section \ref{proof:main} uses $L=1, q\geq 4$, so the rate achieved for arbitrary $K$ is 
\begin{align}
R&=\frac{1}{3}\left(1-\frac{1}{4^K}\right)^{-1}.
\end{align}
However, note that if the field size could be reduced to $q=2$, then the rate achieved by the scheme would become $$\frac{1}{3}\left(1-\frac{1}{2^K}\right)^{-1}=\frac{2}{3}\left(\frac{1-\frac{1}{2}}{1-\frac{1}{2^K}}\right)=\frac{2}{3}\left(1+\frac{1}{2}+\frac{1}{2^2}+\cdots+\frac{1}{2^{K-1}}\right)^{-1}$$ which matches the capacity upper bound from Theorem \ref{theorem:converse}. 
Surprisingly, this can be done with some modification to the structure of the scheme, as explained below.

Suppose each message $W_k, k \in [1:K]$ consists of $L=1$ symbol (bit) from $\mathbb{F}_2$. Let ${\bf W} = (W_1, W_2, \cdots, W_K)$ be a random row vector in $\mathbb{F}_2^{1\times K}$, containing all messages. Let ${\bf Z}$ and ${\bf Z}'$ be uniformly random noise vectors from $\mathbb{F}_2^{1\times K}$ and $\mathbb{F}_2^{K\times 1}$, that are used to guarantee data security and user privacy, respectively. The noise vectors are independent of each other and of the message vector and $\theta$, i.e., $H({\bf W}, {\bf Z}, {\bf Z}', \theta) = H({\bf W}) + H({\bf Z}) + H({\bf Z}') + H(\theta)$. Let ${\bf Q}_\theta$ represent the $\theta^{th}$ column of ${\bf I}_K$ (the $K \times K$ identity matrix). Note that ${\bf W}{\bf Q}_\theta = W_\theta$ is the message desired by the user. The storage at the servers, the queries and the answers are listed below.
\begin{align*}
\begin{array}{|c|c|c|c|}\hline
&\mbox{Server $1$} & \mbox{Server $2$} &\mbox{Server $3$}\\\hline
\mbox{Storage $S_n$} & \mathbf{W}+\mathbf{Z} &\mathbf{W}+\mathbf{Z}\mathbf{B}&\mathbf{Z}\\\hline
\mbox{Query $Q_n^{[\theta]}$} &\mathbf{Z'}&\mathbf{Q_\theta}+\mathbf{Z'}&(\mathbf{I}_K + \mathbf{B})\mathbf{Z'} + \mathbf{B}{\bf Q}_\theta\\\hline
\mbox{Answer $A_n^{[\theta]}$}&\mathbf{W}\mathbf{Z'}+\mathbf{Z}\mathbf{Z}'&\mathbf{W}\mathbf{Q_\theta}+\mathbf{W}\mathbf{Z'}+\mathbf{Z}\mathbf{B}\mathbf{Z}'+\mathbf{Z}\mathbf{B}\mathbf{Q_\theta}&\mathbf{Z}\mathbf{Z}'+\mathbf{Z}\mathbf{B}\mathbf{Z}'+\mathbf{Z}\mathbf{B}\mathbf{Q_\theta}\\\hline
\end{array}
\end{align*}
where ${\bf B}$ is a $K\times K$ deterministic binary matrix such that  ${\bf B}$ and ${\bf I}_K + {\bf B}$ are both full rank. Any such choice of ${\bf B}$ will work for our scheme. The existence of such ${\bf B}$ is established in the following lemma whose proof  appears in Appendix \ref{sec:bk}.
\begin{lemma}\label{lemma:bk}
For all $K\geq 2$, there exists a   matrix ${\bf B}\in\mathbb{F}_2^{K\times K}$ such that ${\bf B}$ and ${\bf I}_K + {\bf B}$ are both invertible. 
\end{lemma}
Now, let us check the correctness, security and privacy of this scheme. The scheme is obviously correct because by adding the three answers shown in the table above, the user recovers $W_\theta$. It is obviously secure because ${\bf B}$ is invertible, so ${\bf Z}{\bf B}\sim{\bf Z}$, is still uniform noise independent of ${\bf W}$. And similarly, it is also obviously private, because ${\bf I}_K+{\bf B}$ is also invertible, so $({\bf I}_K+{\bf B}){\bf Z}' \sim {\bf Z}'$ is still uniform noise independent of ${\bf B}{\bf Q}_\theta$. Thus, surprisingly, we have achieved the capacity of XSTPIR for arbitrary $K$, when $X=T=1$ and $N=3$, completing the proof of Theorem \ref{theorem:N3}. $\hfill\square$

\section{Conclusion}
The  XSTPIR problem is timely due to the growing importance of privacy and security concerns in modern information storage and retrieval systems. It is a conceptually rich topic that reveals new insights into alignment of noise terms, dependence of coding and query structures, cost of symmetric security, significance of field size for the rate of information retrieval, etc. As indicated by various open problems identified here, XSTPIR is a fertile research avenue for future work. In particular, the capacity characterization for arbitrary $K$ could reveal fundamentally new schemes for PIR. Especially intriguing would be the role that field size might play in such a result. Capacity of Sym-XSTPIR is another promising open problem. XSTPIR with constraints on the amount of storage per server, coded storage, multi-message retrieval are other open problems that merit investigation.

\appendix
\section{Proof of Lemma \ref{lemma:bk}}\label{sec:bk}
Let ${\bf J}_k$  denote the $k\times k$ anti-diagonal identity matrix, and let ${\bf 0}_{k_1 \times k_2}$  denote the $k_1 \times k_2$ matrix where all elements are equal to 0 (when $k_1 = k_2$, this notation is further simplified to ${\bf 0}_{k_1}$). Define
\begin{align}
{\bf I}'_{k} = \left[\begin{array}{cc}
{\bf I}_k & {\bf 0}_{k\times 1}\\
{\bf 0}_{1\times k} & 0
\end{array}
\right].
\end{align}

\noindent Choose ${\bf B}$ as follows.

\begin{align}
{\bf B}&=\left\{
\begin{array}{ll}
\left[\begin{array}{cc}
{\bf I}_{\frac{K}{2}} & {\bf J}_{\frac{K}{2}} \\
{\bf J}_{\frac{K}{2}} & {\bf 0}_{\frac{K}{2}}
\end{array}
\right], & \mbox{if } K \mbox{ is even,}\\
\\
 \left[\begin{array}{c|c}
 {\bf J}_{\frac{K+1}{2}} + {\bf I}'_{\frac{K-1}{2}} + {\bf I}_{\frac{K+1}{2}}& \begin{array}{cc}
 {\bf J}_{\frac{K-1}{2}} \\
 {\bf 0}_{1\times\frac{K-1}{2}}
 \end{array} \\ \hdashline
{\bf J}_{\frac{K-1}{2}} ~~ \begin{array}{c}
{\bf 0}_{\frac{K-1}{2}\times 1}
\end{array} & {\bf 0}_{\frac{K-1}{2}}
\end{array}
\right],& \mbox{if } K \mbox{ is odd.}\\
\end{array}
\right. 
\end{align}

\noindent For example,
\begin{flalign}
\mbox{when $K=4$:}&& {\bf B} = \left[
\begin{array}{cccc}
1 & 0 & 0 & 1\\
0 & 1 & 1 & 0\\
0 & 1 & 0 & 0\\
1 & 0 & 0 & 0
\end{array}
\right],
&& \mbox{and when $K=5$:}&&
{\bf B} = \left[
\begin{array}{ccccc}
0 & 0 & 1 & 0 & 1\\
0 & 1 & 0 & 1 & 0\\
1 & 0 & 1 & 0 & 0\\
0 & 1 & 0 & 0 & 0\\
1 & 0 & 0 & 0 & 0
\end{array}
\right].
\end{flalign}

Let us show that ${\bf B}$ and ${\bf I}_K + {\bf B}$ are both invertible.

First, consider ${\bf B}$. Regardless of whether $K$ is even or odd, ${\bf B}$ is an upper anti-triangular matrix where all anti-diagonal elements are 1 so that $\det({\bf B}) = 1$ and ${\bf B}$ has full rank.

Next, consider ${\bf I}_K + {\bf B}$.
\begin{align}
\mbox{When $K$ is even:} & ~~{\bf I}_K + {\bf B} = \left[\begin{array}{cc}
{\bf I}_{\frac{K}{2}} & {\bf 0}_{\frac{K}{2}} \\
{\bf 0}_{\frac{K}{2}} & {\bf I}_{\frac{K}{2}}
\end{array}
\right]
+ \left[\begin{array}{cc}
{\bf I}_{\frac{K}{2}} & {\bf J}_{\frac{K}{2}} \\
{\bf J}_{\frac{K}{2}} & {\bf 0}_{\frac{K}{2}}
\end{array}
\right] = 
\left[\begin{array}{cc}
{\bf 0}_{\frac{K}{2}} & {\bf J}_{\frac{K}{2}} \\
{\bf J}_{\frac{K}{2}} & {\bf I}_{\frac{K}{2}}
\end{array}
\right] \\
&~~\Rightarrow ~~\det({\bf I}_K+{\bf B}) = 1. \\
\mbox{When $K$ is odd:} & ~~{\bf I}_K + {\bf B} = \left[\begin{array}{c|c}
{\bf I}_{\frac{K+1}{2}} & {\bf 0}_{\frac{K+1}{2} \times \frac{K-1}{2}} \\ \hdashline
{\bf 0}_{\frac{K-1}{2} \times \frac{K+1}{2}} & {\bf I}_{\frac{K-1}{2}}
\end{array}
\right]
+   \left[\begin{array}{c|c}
 {\bf J}_{\frac{K+1}{2}} + {\bf I}'_{\frac{K-1}{2}} + {\bf I}_{\frac{K+1}{2}}& \begin{array}{cc}
 {\bf J}_{\frac{K-1}{2}} \\
 {\bf 0}_{1\times\frac{K-1}{2}}
 \end{array} \\ \hdashline
{\bf J}_{\frac{K-1}{2}} ~~ \begin{array}{c}
{\bf 0}_{\frac{K-1}{2}\times 1}
\end{array} & {\bf 0}_{\frac{K-1}{2}}
\end{array}
\right] \\
&~~~~~~~~~~~= \left[\begin{array}{c|c}
 {\bf J}_{\frac{K+1}{2}} + {\bf I}'_{\frac{K-1}{2}}  & \begin{array}{cc}
 {\bf J}_{\frac{K-1}{2}} \\
 {\bf 0}_{1\times\frac{K-1}{2}}
 \end{array} \\ \hdashline
{\bf J}_{\frac{K-1}{2}} ~~ \begin{array}{c}
{\bf 0}_{\frac{K-1}{2}\times 1}
\end{array} & {\bf I}_{\frac{K-1}{2}}
\end{array}
\right] \\
&~~\Rightarrow ~~\det({\bf I}_K+{\bf B}) = \det\left( {\bf J}_{\frac{K+1}{2}} + {\bf I}'_{\frac{K-1}{2}}  + \left[\begin{array}{cc}
 {\bf J}_{\frac{K-1}{2}} \\
 {\bf 0}_{1\times\frac{K-1}{2}}
 \end{array}\right] {\bf I}_{\frac{K-1}{2}}^{-1} \left[
 {\bf J}_{\frac{K-1}{2}} ~~ \begin{array}{c}
{\bf 0}_{\frac{K-1}{2}\times 1}
\end{array}
 \right] 
 \right) \label{eq:e4}\\
 &~~~~~~~~~~~~~~~~~~~~~~~~~= \det\left( {\bf J}_{\frac{K+1}{2}} + {\bf I}'_{\frac{K-1}{2}}  +  {\bf I}'_{\frac{K-1}{2}} \right) = \det\left( {\bf J}_{\frac{K+1}{2}} \right) = 1
\end{align}
where (\ref{eq:e4}) follows from the following formula on the determinant of a block matrix that is made up of matrices ${\bf A}, {\bf B}, {\bf C}, {\bf D}$ with proper dimensions and ${\bf D}$ is invertible.
\begin{equation}
{\displaystyle {\rm {det}}\left({\begin{matrix}\mathbf{A}&\mathbf{B}\\\mathbf{C}&\mathbf{D}\end{matrix}}\right)={\rm {det}}(\mathbf{D})\,{\rm {det}}\left(\mathbf{A}-\mathbf{BD}^{-1}\mathbf{C}\right).}
\end{equation}
The proof is thus complete.$\hfill\square$

\section{Proof of Corollaries \ref{cor:converse}, \ref{cor:mds}, \ref{cor:main}, \ref{cor:N3}}\label{proof:cor}
The proof of Corollary \ref{cor:converse} is trivial because imposing the symmetric security constraint cannot increase capacity.

\subsection{Proof of Corollary \ref{cor:mds}}
To prove Corollary \ref{cor:mds}  we provide a scheme as follows. Each message $W_k, k\in[1:K]$, consists of $L=1$ symbol from some finite field $\mathbb{F}_q$. Let $Z_{x,k,m}, x\in[1:X],  k\in[1:K], m\in[1:K]$ be independent uniform noise symbols from $\mathbb{F}_q$. The subscript, $m$, in $Z_{x,k,m}$ is interpreted modulo $K$, i.e., $Z_{x, k, m}=Z_{x, k,m+K}$.
The storage at each server is specified as,
\begin{align}
S_n&= \{Z_{n,k,m}, k\in[1:K], m\in[1:K]\}, && n\in[1:X],\\
S_n&= \left\{W_{k}+\sum_{x=1}^X Z_{x,k,m}, k\in[1:K], m\in[1:K]\right\}, && n=N.
\end{align}
The queries from each server are specified as,
\begin{align}
Q^{[\theta]}_n&: \mbox{Ask for } \{Z_{n,k, m_o}, k\in[1:K]\},&& n\in[1:X],\\
Q^{[\theta]}_n&: \mbox{Ask for } \{W_{k}+\sum_{x=1}^X Z_{x,k,m_o-\theta+k}, k\in[1:K]\},&&n=N,
\end{align}
where $m_o$ is chosen privately and uniformly randomly by the user from $[1:K]$.
Thus, in order to retrieve $1$ desired message symbol, the user downloads a total of $KN$ symbols from all servers. The scheme is $X$-secure because each message symbol is protected by independent uniform noise terms. It is correct because for $k=\theta$ the download from Server $N$, contains the symbol $W_{\theta}+\sum_{x=1}^XZ_{x,\theta,m_o}$ and the downloads from the first $X$ servers include all the noise terms $Z_{x, \theta, m_o}$.  The scheme is private because $m_o$ is chosen uniformly and privately by the user. It satisfies symmetric security because all the undesired message symbols $W_{k}, k\neq \theta$, contained in the answers are protected by noise terms $Z_{x,k,m_o-\theta+k}$ and these noise terms are independent of the noise terms downloaded from servers $n\in[1:X]$ because $m_o-\theta+k\neq m_o$ when $k\neq \theta$. The rate achieved is $\frac{1}{KN}$, which is the capacity for this setting.$\hfill\square$

Note that in  the Sym-XSPIR scheme described above, each server stores $K^2$ symbols, when the total data is only $KL=K$ symbols. Thus, this Sym-XSPIR scheme takes advantage of unconstrained storage when $K$ is large, more so than the XSTPIR schemes which store no more than $KL$ symbols at each server.

\subsection{Proof of Corollary \ref{cor:main}}
To prove Corollary \ref{cor:main}, we show that the scheme presented in Section \ref{proof:main} automatically guarantees symmetric security when $T=1$. Define
\begin{align}
{\bf W}_{i}^c&=\{{\bf W}_{l}, l\in[1:L],   l\neq i \}\\
{\bf Z}_{ij}^c&=\{{\bf Z}_{lx}, l\in[1:L],  x\in[1:X], (l,x)\neq(i,j) \}.
\end{align}
We need to prove that beyond the information that the user must have, i.e., $W_\theta, Q_{[1:N]}^{[\theta]},\theta$, he cannot learn anything about the messages ${\bf W}_{[1:L]}$ from the answers $A_{[1:N]}^{[\theta]}$.
\begin{align}
I\left({\bf W}_{[1:L]}; A_{[1:N]}^{[\theta]}\mid W_{\theta}, Q^{[\theta]}_{[1:N]},\theta\right)&=\sum_{l\in[1:L]}I\left({\bf W}_l; A_{[1:N]}^{[\theta]}\mid {\bf W}_{[1:l-1]}, W_{\theta}, Q^{[\theta]}_{[1:N]},\theta\right)\\
&\leq \sum_{l\in[1:L]}I\left({\bf W}_l; A_{[1:N]}^{[\theta]}\mid {\bf W}_{l}^c, W_{\theta}, Q^{[\theta]}_{[1:N]},\theta\right)\\
&\leq \sum_{l\in[1:L]} I\left({\bf W}_l; A_{[1:N]}^{[\theta]}\mid {\bf Z}_{l1}^c, {\bf W}_{l}^c, W_{\theta}, Q^{[\theta]}_{[1:N]},\theta\right)
\end{align}
where we repeatedly used the fact that $I(A;B\mid C)\leq I(A;B\mid C,D)$ if $I(A;D\mid C)=0$ and the facts that
\begin{align}
I({\bf W}_l; {\bf W}_{[l+1:L]}\mid W_{[1:l-1]},W_\theta,Q_{[1:N]}^{[\theta]},\theta)&=0\\
I\left({\bf W}_l;  {\bf Z}_{l1}^c\mid {\bf W}_{l}^c, W_{\theta}, Q^{[\theta]}_{[1:N]},\theta\right)&=0
\end{align} 
that follow from the indepenence of messages, queries, and the noise terms, by construction of the scheme in Section \ref{proof:main}.
To prove  Corollary \ref{cor:main} it suffices to show that each of the terms in the summation is zero. Without loss of generality, let us consider $l=1$. Because of the conditioning on ${\bf Z}_{11}^c, {\bf W}_{1}^c, W_{\theta}, Q^{[\theta]}_{[1:N]},\theta$, we can subtract the contributions from these terms, whose values are fixed, from  $A_n^{[\theta]}$, leaving us with only
\begin{align}
{A'}^{[\theta]}_n&=\left({\bf W}_1+(1+\alpha_n){\bf Z}_{11}\right)\left(\frac{\Delta_n}{1+\alpha_n}\right)\left({\bf Q}_\theta+(1+\alpha_n){\bf Z}'_{1}\right)\\
&=\left(\frac{\Delta_n}{1+\alpha_n}\right){\bf W}_1{\bf Q}_{\theta}+\Delta_n({\bf W}_1{\bf Z}'_{1}+{\bf Z}_{11}{\bf Q}_\theta)+\Delta_n(1+\alpha_n){\bf Z}_{11}{\bf Z}_{1}'\\
&=\left(\frac{\Delta_n}{1+\alpha_n}\right)W_{\theta 1}+\Delta_n({\bf W}_1{\bf Z}'_{1}+{\bf Z}_{11}(\theta))+\Delta_n(1+\alpha_n){\bf Z}_{11}{\bf Z}_{1}'\label{eq:thus}
\end{align}
where ${\bf Z}_{11}(i)$ is the $i^{th}$ element of the vector ${\bf Z}_{11}$. Note that $W_{\theta 1}$ is also a constant because of the conditioning on $W_\theta$.
Given ${\bf Z}_{11}^c, {\bf W}_{1}^c, W_{\theta}, Q^{[\theta]}_{[1:N]},\theta$, the random variable $A_{[1:N]}^{[\theta]}$ is  an invertible function of ${A'}^{[\theta]}_{[1:N]}$. 
\begin{align}
&I\left({\bf W}_1; A_{[1:N]}^{[\theta]}\mid {\bf Z}_{11}^c, {\bf W}_{1}^c, W_{\theta}, Q^{[\theta]}_{[1:N]},\theta\right)\\
&= I\left({\bf W}_1; {A'}_{[1:N]}^{[\theta]}\mid {\bf Z}_{11}^c, {\bf W}_{1}^c, W_{\theta}, Q^{[\theta]}_{[1:N]},\theta\right)\\
&= I\left({\bf W}_1; {A'}_{[1:N]}^{[\theta]}\mid  W_{\theta 1}, Q^{[\theta]}_{[1:N]},\theta\right)\\
&\leq I({\bf W}_1; {\bf W}_1{\bf Z}_{1}'+{\bf Z}_{11}(\theta),  {\bf Z}_{11}{\bf Z}_{1}'\mid  W_{\theta 1}, {\bf Q}_\theta,  \theta)\label{eq:explain1}\\
&\leq I({\bf W}_1; {\bf W}_1{\bf Z}_{1}'+{\bf Z}_{11}(\theta),  {\bf Z}_{11}{\bf Z}_{1}'\mid  W_{\theta 1}, {\bf Q}_\theta, {\bf Z}_{1}', \theta)\label{eq:explain2}\\
&=0.
\end{align}
In \eqref{eq:explain1} we used the fact that given $W_{\theta 1}$, the random variable ${A'}_{[1:N]}^{[\theta]}$ is a function of ${\bf W}_1{\bf Z}_{1}'+{\bf Z}_{11}(\theta),  {\bf Z}_{11}{\bf Z}_{1}'$ because of \eqref{eq:thus}, and the fact that for any random variables $A,B,C$, we must have $I(A;f(B)\mid C)\leq I(A;B\mid C)$. In \eqref{eq:explain2} we used the fact that conditioning on an independent random variable cannot reduce mutual information, i.e., $I(A;B\mid C)\leq I(A;B\mid C,D)$ if $I(A;D\mid C)=0$, and the fact that ${\bf Z}_1'$ is independent of ${\bf W}_1$ after conditioning on $W_{\theta 1}, {\bf Q}_\theta, \theta$ by construction of the scheme as described in Section \ref{proof:main}.
The last step is justified as follows. Because of the conditioning on ${\bf Z}_1'$, its value is a  constant for which there are only three possibilities: ${\bf Z}_1'$ is either the zero vector, or it is equal to $\mu{\bf Q}_\theta$ for some non-zero $\mu\in\mathbb{F}_q$, or it is neither zero nor equal to $\mu{\bf Q}_\theta$. If ${\bf Z}_1'$ is the zero vector, then the mutual information is automatically zero because ${\bf W}_1$ is eliminated entirely. If ${\bf Z}_1'=\mu{\bf Q}_\theta$ for some non-zero $\mu$, then ${\bf W}_1{\bf Z}_1'=\mu W_{\theta 1}$ and the mutual information is again zero because of the conditioning on $W_{\theta 1}$. Finally, if ${\bf Z}_1'$ is neither zero nor a scaled version of ${\bf Q}_\theta$, then ${\bf Z}_{11}{\bf Z}_1'$ is a sum of uniformly random noise terms  in $\mathbb{F}_q$, at least one of which is independent of ${\bf Z}_{11}(\theta)$ and ${\bf Z}_1'$. So in this case also the mutual information is zero. This completes the proof of Corollary \ref{cor:main}. $\hfill\square$

\subsection{Proof of Corollary \ref{cor:N3}}
The proof of Corollary \ref{cor:N3} is presented next. Recall that in the scheme of the proof of  Theorem \ref{theorem:N3}, the user obtains the following three symbols from the answers,
\begin{eqnarray}
&& \mathbf{W}\mathbf{Q_\theta} = W_\theta\\
&& \mathbf{W}\mathbf{Z}'+\mathbf{Z}\mathbf{Z}' \label{eq:syma2}\\
&& \mathbf{W}\mathbf{Z}'+\mathbf{Z}\mathbf{B}\mathbf{Z}'+\mathbf{Z}\mathbf{B}\mathbf{Q}_\theta \label{eq:syma3}.
\end{eqnarray}
We show that symmetric security holds, i.e., conditioned on ${\bf Z}'$, from these three symbols the user learns nothing about the undesired messages $W_1,\cdots, W_{\theta-1}, W_{\theta+1}, \cdots, W_K$. When $\mathbf{Z}'$ is the zero vector, the symbol $\mathbf{W}\mathbf{Z}'$ is zero as well, leaking nothing about the  undesired messages.  Now consider (\ref{eq:syma2}). If $\mathbf{Z}'$ is not the zero vector, then the symbol $\mathbf{W}\mathbf{Z}'$ is protected by an independent noise term.  Similarly, consider (\ref{eq:syma3}) and consider three possibilities: ${\bf B(Z'+Q_\theta)}$ is either  zero, or equal to ${\bf Z}'$, or not zero and not equal to ${\bf Z}'$. If $\mathbf{B}(\mathbf{Z}' + {\bf Q}_\theta)$ is the zero vector, then because ${\bf B}$ is invertible, we must have ${\bf Z}' = {\bf Q}_\theta$, so the symbol $\mathbf{W}\mathbf{Z}' = \mathbf{W}\mathbf{Q}_\theta$ is the desired message, again leaking nothing about undesired messages. If $\mathbf{B}(\mathbf{Z}' + {\bf Q}_\theta)={\bf Z}'$ then \eqref{eq:syma3} is redundant, i.e., same as \eqref{eq:syma2}, so it leaks no new information. Finally, if $\mathbf{B}(\mathbf{Z}' + {\bf Q}_\theta)$ is not zero and not equal to ${\bf Z}'$, then ${\bf ZB(Z'+Q_\theta)}$ is independent of ${\bf ZZ}'$, so that \eqref{eq:syma3} is protected by an independent noise term.
Therefore, in all cases, the user learns nothing about undesired messages, and this completes the proof of symmetric security.$\hfill\square$

\medskip

\bibliographystyle{IEEEtran}
\bibliography{Thesis}

\begin{thebibliography}{10}
\providecommand{\url}[1]{#1}
\csname url@samestyle\endcsname
\providecommand{\newblock}{\relax}
\providecommand{\bibinfo}[2]{#2}
\providecommand{\BIBentrySTDinterwordspacing}{\spaceskip=0pt\relax}
\providecommand{\BIBentryALTinterwordstretchfactor}{4}
\providecommand{\BIBentryALTinterwordspacing}{\spaceskip=\fontdimen2\font plus
\BIBentryALTinterwordstretchfactor\fontdimen3\font minus
  \fontdimen4\font\relax}
\providecommand{\BIBforeignlanguage}[2]{{%
\expandafter\ifx\csname l@#1\endcsname\relax
\typeout{** WARNING: IEEEtran.bst: No hyphenation pattern has been}%
\typeout{** loaded for the language `#1'. Using the pattern for}%
\typeout{** the default language instead.}%
\else
\language=\csname l@#1\endcsname
\fi
#2}}
\providecommand{\BIBdecl}{\relax}
\BIBdecl

\bibitem{Sun_Jafar_PIR}
H.~Sun and S.~A. Jafar, ``{The Capacity of Private Information Retrieval},''
  \emph{IEEE Transactions on Information Theory}, vol.~63, no.~7, pp.
  4075--4088, July 2017.

\bibitem{Sun_Jafar_TPIR}
------, ``{The Capacity of Robust Private Information Retrieval with Colluding
  Databases},'' \emph{IEEE Transactions on Information Theory}, vol.~64, no.~4,
  pp. 2361--2370, April 2018.

\bibitem{Tajeddine_Gnilke_Karpuk_Etal}
R.~Tajeddine, O.~W. Gnilke, D.~Karpuk, R.~Freij-Hollanti, C.~Hollanti, and
  S.~El~Rouayheb, ``Private information retrieval schemes for codec data with
  arbitrary collusion patterns,'' \emph{IEEE International Symposium on
  Information Theory (ISIT)}, pp. 1908--1912, 2017.

\bibitem{Jia_Sun_Jafar}
Z.~Jia, H.~Sun, and S.~Jafar, ``The capacity of private information retrieval
  with disjoint colluding sets,'' in \emph{IEEE GLOBECOM}, 2017.

\bibitem{Banawan_Ulukus}
K.~Banawan and S.~Ulukus, ``{The Capacity of Private Information Retrieval from
  Coded Databases},'' \emph{IEEE Transactions on Information Theory}, vol.~64,
  no.~3, pp. 1945--1956, 2018.

\bibitem{FREIJ_HOLLANTI}
R.~Freij-Hollanti, O.~Gnilke, C.~Hollanti, and D.~Karpuk, ``{Private
  Information Retrieval from Coded Databases with Colluding Servers},''
  \emph{SIAM Journal on Applied Algebra and Geometry}, vol.~1, no.~1, pp.
  647--664, 2017.

\bibitem{Sun_Jafar_MDSTPIR}
H.~Sun and S.~A. Jafar, ``{Private Information Retrieval from MDS Coded Data
  with Colluding Servers: Settling a Conjecture by Freij-Hollanti et al.}''
  \emph{IEEE Transactions on Information Theory}, vol.~64, no.~2, pp.
  1000--1022, February 2018.

\bibitem{Lin_Kumar_Rosnes_Amat}
H.-Y. Lin, S.~Kumar, E.~Rosnes, and A.~G. i~Amat, ``A capacity-achieving {PIR}
  protocol for distributed storage using an arbitrary linear code,''
  \emph{arXiv preprint arXiv:1801.04923}, 2018.

\bibitem{Attia_Kumar_Tandon}
R.~T. Mohamed Adel~Attia, Deepak~Kumar, ``The capacity of private information
  retrieval from uncoded storage constrained databases,'' \emph{arXiv preprint
  arXiv:1805.04104}, 2018.

\bibitem{Banawan_Ulukus_MPIR}
K.~Banawan and S.~Ulukus, ``Multi-message private information retrieval:
  Capacity results and near-optimal schemes,'' \emph{arXiv preprint
  arXiv:1702.01739}, 2017.

\bibitem{Sun_Jafar_PIRL}
H.~Sun and S.~A. Jafar, ``Optimal download cost of private information
  retrieval for arbitrary message length,'' \emph{IEEE Transactions on
  Information Forensics and Security}, vol.~12, no.~12, pp. 2920--2932, 2017.

\bibitem{Sun_Jafar_MPIR}
------, ``{Multiround Private Information Retrieval: Capacity and Storage
  Overhead},'' \emph{IEEE Transactions on Information Theory}, vol.~64, no.~8,
  pp. 5743--5754, August 2018.

\bibitem{Tian_Sun_Chen}
C.~Tian, H.~Sun, and J.~Chen, ``A {Shannon}-theoretic approach to the
  storage-retrieval tradeoff in pir systems,'' \emph{IEEE International
  Symposium on Information Theory, ISIT 2018}, pp. 1904--1908, 2018.

\bibitem{Kadhe_Garcia_Heidarzadeh_Rouayheb_Sprintson}
S.~Kadhe, B.~Garcia, A.~Heidarzadeh, S.~E. Rouayheb, and A.~Sprintson,
  ``Private information retrieval with side information,'' \emph{arXiv preprint
  arXiv:1709.00112}, 2017.

\bibitem{Chen_Wang_Jafar_Side}
Z.~Chen, Z.~Wang, and S.~Jafar, ``The capacity of private information retrieval
  with private side information,'' \emph{arXiv preprint arXiv:1709.03022},
  2017.

\bibitem{Tandon_CachePIR}
R.~Tandon, ``The capacity of cache aided private information retrieval,''
  \emph{arXiv preprint arXiv:1706.07035}, 2017.

\bibitem{Wei_Banawan_Ulukus}
Y.-P. Wei, K.~Banawan, and S.~Ulukus, ``Fundamental limits of cache-aided
  private information retrieval with unknown and uncoded prefetching,''
  \emph{arXiv preprint arXiv:1709.01056}, 2017.

\bibitem{Wei_Banawan_Ulukus_Side}
------, ``The capacity of private information retrieval with partially known
  private side information,'' \emph{arXiv preprint arXiv:1710.00809}, 2017.

\bibitem{Shariatpanahi_Siavoshani_Maddah}
S.~P. Shariatpanahi, M.~J. Siavoshani, and M.~A. Maddah-Ali, ``Multi-message
  private information retrieval with private side information,'' \emph{arXiv
  preprint arXiv:1805.11892}, 2018.

\bibitem{Li_Gastpar}
S.~Li and M.~Gastpar, ``Single-server multi-message private information
  retrieval with side information,'' \emph{arXiv preprint arXiv:1808.05797},
  2018.

\bibitem{Banawan_Ulukus_Asymmetric}
K.~Banawan and S.~Ulukus, ``Private information retrieval through wiretap
  channel ii: Privacy meets security,'' \emph{arXiv preprint arXiv:1801.06171},
  2018.

\bibitem{Wang_Sun_Skoglund}
Q.~Wang, H.~Sun, and M.~Skoglund, ``The capacity of private information
  retrieval with eavesdroppers,'' \emph{arXiv preprint arXiv:1804.10189}, 2018.

\bibitem{Banawan_Ulukus_Byzantine}
K.~Banawan and S.~Ulukus, ``The capacity of private information retrieval from
  byzantine and colluding databases,'' \emph{arXiv preprint arXiv:1706.01442},
  2017.

\bibitem{Zhang_Ge_Variant}
Y.~Zhang and G.~Ge, ``Private information retrieval from {MDS} coded databases
  with colluding servers under several variant models,'' \emph{arXiv preprint
  arXiv:1705.03186}, 2017.

\bibitem{Tajeddine_Gnilke_Karpuk_Hollanti}
R.~Tajeddine, O.~W. Gnilke, D.~Karpuk, R.~Freij-Hollanti, and C.~Hollanti,
  ``Private information retrieval from coded storage systems with colluding,
  byzantine, and unresponsive servers,'' \emph{arXiv preprint
  arXiv:1806.08006}, 2018.

\bibitem{wang2018e}
Q.~Wang, H.~Sun, and M.~Skoglund, ``The $\epsilon$-error capacity of symmetric
  pir with byzantine adversaries,'' in \emph{2018 IEEE Information Theory
  Workshop (ITW)}.\hskip 1em plus 0.5em minus 0.4em\relax IEEE, 2018, pp. 1--5.

\bibitem{yao2019capacity}
X.~Yao, N.~Liu, and W.~Kang, ``The capacity of multi-round private information
  retrieval from byzantine databases,'' \emph{arXiv preprint arXiv:1901.06907},
  2019.

\bibitem{Sun_Jafar_SPIR}
H.~Sun and S.~A. Jafar, ``The capacity of symmetric private information
  retrieval,'' \emph{IEEE Transactions on Information Theory}, 2018.

\bibitem{Wang_Skoglund_TSPIR}
Q.~Wang and M.~Skoglund, ``Linear symmetric private information retrieval for
  {MDS} coded distributed storage with colluding servers,'' \emph{arXiv
  preprint arXiv:1708.05673}, 2017.

\bibitem{Wang_Skoglund_SPIRAd}
------, ``Secure symmetric private information retrieval from colluding
  databases with adversaries,'' \emph{arXiv preprint arXiv:1707.02152}, 2017.

\bibitem{Yang_Shin_Lee}
H.~Yang, W.~Shin, and J.~Lee, ``Private information retrieval for secure
  distributed storage systems,'' \emph{IEEE Transactions on Information
  Forensics and Security}, vol.~13, no.~12, pp. 2953--2964, December 2018.

\bibitem{Sun_Jafar_BIAPIR}
H.~Sun and S.~A. Jafar, ``Blind interference alignment for private information
  retrieval,'' \emph{2016 IEEE International Symposium on Information Theory
  (ISIT)}, pp. 560--564, 2016.

\bibitem{Sun_Jafar_PC}
------, ``The capacity of private computation,'' \emph{arXiv preprint
  arXiv:1710.11098}, 2017.

\bibitem{Mirmohseni_Maddah}
M.~Mirmohseni and M.~A. Maddah-Ali, ``Private function retrieval,'' \emph{arXiv
  preprint arXiv:1711.04677}, 2017.

\end{thebibliography}
\end{document}